\documentclass[11pt]{article}
\usepackage{amsthm,amssymb,amsmath,amsfonts}
\usepackage[utf8x]{inputenc}
\usepackage[margin=1in]{geometry}
\usepackage{latexsym}
\usepackage{enumerate}
\usepackage{graphicx}
\usepackage{enumitem}
\usepackage{color}
\definecolor{DarkBlue}{RGB}{0,0,150}
\usepackage[colorlinks=true,linkcolor=DarkBlue,citecolor=DarkBlue]{hyperref}
\usepackage{xcolor}
\usepackage{framed}
\usepackage{mathrsfs}
\usepackage[mathscr]{euscript}
\usepackage{soul}
\usepackage{xspace}
\usepackage{mathtools}

\theoremstyle{plain}
\newtheorem{theorem}{Theorem}
\newtheorem{proposition}[theorem]{Proposition}
\newtheorem{lemma}[theorem]{Lemma}
\newtheorem{corollary}[theorem]{Corollary}

\newtheorem{conjecture}{Conjecture}
\newtheorem{openq}{Open Question}

\theoremstyle{definition}
\newtheorem{definition}[theorem]{Definition}
\newtheorem{construction}[theorem]{Construction}

\theoremstyle{remark}
\newtheorem{remark}{Remark}

\newcommand{\NN}{\mathbb{N}}
\newcommand{\ZZ}{\mathbb{Z}}
\newcommand{\RR}{\mathbb{R}}

\newcommand{\eps}{\varepsilon}

\renewcommand{\mod}{\bmod}

\DeclareMathOperator{\poly}{poly}

\newcommand{\Enc}{\mathsf{Enc}}
\newcommand{\Dec}{\mathsf{Dec}}

\newcommand{\cA}{{\cal A}}
\newcommand{\cB}{{\cal B}}

\newcommand{\Index}{\mathsf{Index}}
\newcommand{\Query}{\mathsf{Query}}

\DeclareMathAlphabet{\mathbfsf}{\encodingdefault}{\sfdefault}{bx}{n}

\newcommand{\sdist}[1]{\| #1 \|_{\sf s}}

\newcommand{\zo}{\{0,1\}}

\newcommand{\ThreeXOR}{\textsf{3XOR}\xspace}
\newcommand{\ThreeSUM}{\textsf{3SUM}\xspace}

\newcommand{\ThreeXORInd}{\textsf{3XOR-Indexing}\xspace}
\newcommand{\ThreeSUMInd}{\textsf{3SUM-Indexing}\xspace}
\newcommand{\kSUM}{\textsf{kSUM}\xspace}
\newcommand{\kSUMInd}{\textsf{kSUM-Indexing}\xspace}

\newcommand{\X}{\mathcal{X}}

\def\SUM{\ThreeSUM}
\def\SUMInd{\ThreeSUMInd}

\title{Data Structures Meet Cryptography: \\ 3SUM with Preprocessing}
\author{Alexander Golovnev\\Harvard\\\texttt{alexgolovnev@gmail.com} \and Siyao Guo\\NYU Shanghai\\\texttt{sg191@nyu.edu} \and Thibaut Horel\\MIT\\\texttt{thibauth@mit.edu} \and Sunoo Park\\MIT \& Harvard\\\texttt{sunoo@csail.mit.edu} \and Vinod Vaikuntanathan\\MIT\\\texttt{vinodv@csail.mit.edu}}
\date{}

\begin{document}
\maketitle

\begin{abstract}
  This paper shows several connections between
  data structure problems and cryptography against preprocessing attacks.
  Our results span data structure upper bounds,
  cryptographic applications, and data structure lower bounds,
  as summarized next.

  First, we apply Fiat--Naor inversion, a technique with cryptographic origins,
  to obtain a data structure upper bound.
  In particular, our technique yields a suite of algorithms with space $S$ and (online) time $T$
  for a preprocessing version of the $N$-input $\SUM$ problem
  where $S^3\cdot T = \widetilde{O}(N^6)$.
  This disproves a strong conjecture (Goldstein \emph{et al.}, WADS 2017)
  that there is no data structure that solves this problem
  for $S=N^{2-\delta}$ and $T = N^{1-\delta}$ for any constant $\delta>0$.

  Secondly, we show equivalence between lower bounds for a broad class of (static) data structure problems and
one-way functions in the random oracle model
  that resist a very strong form of preprocessing attack.
  Concretely, given a random function
  $F: [N] \to [N]$ (accessed as an oracle) we show how to {\em compile} it into
  a function $G^F: [N^2] \to [N^2]$ which resists $S$-bit preprocessing attacks
  that run in query time $T$ where $ST=O(N^{2-\eps})$
  (assuming a corresponding data structure lower bound on $\SUM$).
  In contrast, a classical result of Hellman tells us that $F$ itself can be
  more easily inverted, say with $N^{2/3}$-bit preprocessing in $N^{2/3}$ time.
  We also show that much stronger lower bounds follow from the hardness of
  $\kSUM$.
  Our results can be equivalently interpreted as
  security against adversaries that are very non-uniform, or have large auxiliary input,
  or as security in the face of a powerfully backdoored
  random oracle.

  Thirdly, we give non-adaptive lower bounds for $\SUM$
  which match the best known lower bounds for static data structure problems.
  Moreover, we show that our lower bound generalizes to a range of geometric problems,
  such as three points on a line, polygon containment, and others.
\end{abstract}

\vfill
\thispagestyle{empty}
\newpage
\tableofcontents
\thispagestyle{empty}
\newpage
\pagenumbering{arabic}

\section{Introduction}

\def\Omegat{\widetilde{\Omega}}

Cryptography and data structures have long enjoyed a productive
relationship~\cite{Hel80,fn00,LN93,NSW08,BHKN13,BN16,ANSS16,NY15,LN18,JLN19}: indeed,
the relationship has been referred to as a ``match made in heaven''~\cite{N13}.
In this paper, we initiate the study of a new connection between the two
fields, which allows us to construct novel cryptographic objects starting from
data structure lower bounds, and vice versa. Our results are three-fold. Our
first result is a new upper bound for a data structure version of the classical
$\SUM$ problem (called $\ThreeSUMInd$) using Fiat--Naor inversion~\cite{fn00}, a technique with cryptographic origins. This result refutes a strong conjecture due to Goldstein, Kopelowitz, Lewenstein and Porat~\cite{GKLP17}.
In our second and main result, we turn this connection around, and show a framework for constructing {\em one-way functions} in the random oracle model whose security bypasses known time/space tradeoffs, relying on any of a broad spectrum of (conjectured) data structure lower bounds (including for $\ThreeSUMInd$). As a third result, we show new lower bounds for a variety of data structure problems (including for $\ThreeSUMInd$) which match the state of the art in the field of static data structure lower bounds.

Next, we describe our results, focusing on the important special case of $\ThreeSUMInd$;
all of our results and methods extend to the more general $\kSUMInd$
problem where pairwise sums are replaced with $(k-1)$-wise sums for an
arbitrary constant integer $k$ independent of the input length.
Section~\ref{sec:threesum-bg} gives background on $\ThreeSUMInd$,
then Section~\ref{sec:intro-results} discusses our contributions.

\subsection{$\ThreeSUM$ and $\ThreeSUMInd$}\label{sec:threesum-bg}

One of the many equivalent formulations of the \ThreeSUM problem is the
following: given a set $A$ of $N$ integers, output $a_1,a_2,a_3\in A$ such that $a_1+a_2=a_3$.
There is an easy $O(N^2)$ time deterministic algorithm for $\ThreeSUM$. Conversely, the popular $\ThreeSUM$ conjecture states that there are no sub-quadratic algorithms for this problem~\cite{go95,e99}.

\begin{conjecture}[The ``Modern \ThreeSUM conjecture'']
\label{conj:3sum}
\ThreeSUM cannot be solved in time $O(N^{2-\delta})$ for any constant $\delta> 0$.
\end{conjecture}

Conjecture~\ref{conj:3sum} has been helpful for understanding the precise hardness of many geometric problems~\cite{go95,dbdgo97,bvkt98,achmssy98,e99,bhp01,ahikllps01,seo03,aek05,ehpm06,cehp07,ahp08,aadhru12}. Furthermore, starting with the works of \cite{vww09,p10}, the \ThreeSUM conjecture has also been used for conditional lower bounds for many combinatorial~\cite{avw14,gklp14,kpp16} and string search problems~\cite{chc09,bcpps13,avww14,acll14,aklpps16,kpp16}.

Our main results relate to a preprocessing variant of $\ThreeSUM$ known as $\ThreeSUMInd$, which was first defined by Demaine and Vadhan~\cite{DV01} in an unpublished note and then by Goldstein, Kopelowitz, Lewenstein and Porat~\cite{GKLP17}. 
In $\ThreeSUMInd$, there is an offline phase where a computationally {\em unbounded} algorithm receives $A = \{a_1,\ldots,a_N\}$ and produces a data structure with $S$ words of $w$ bits each; and an online phase which is given the target $b$ and needs to find a pair $(a_i,a_j)$  such that $a_i+a_j=b$ by probing only $T$ memory cells of the data structure (\emph{i.e.}, taking ``query time'' $T$). The online phase does not receive the set $A$ directly, and there is no bound on the computational complexity of the online phase, only the number of queries it makes.

There are two simple algorithms that solve $\ThreeSUMInd$.
The first stores a sorted version of $A$ as the data structure (so $S=N$) and in the online phase, solves $\ThreeSUMInd$ in $T=O(N)$ time using the standard two-finger algorithm for $\ThreeSUM$. The second stores all pairwise sums of $A$, sorted, as the data structure (so $S=O(N^2)$) and in the online phase, looks up the target $b$ in $T=\widetilde{O}(1)$ time.\footnote{The notation $\widetilde{O}(f(N))$ suppresses poly-logarithmic factors in $f(N)$.} There were no other algorithms known prior to this work. This led \cite{DV01,GKLP17} to formulate the following three conjectures.

\begin{conjecture}[\cite{GKLP17}]
	\label{conj:weakintro}
	If there exists an algorithm which solves $\SUMInd$ with
	preprocessing space $S$ and $T = \widetilde O(1)$ probes
	then $S=\widetilde\Omega(N^2)$.
\end{conjecture}

\begin{conjecture}[\cite{DV01}]
	\label{conj:dvintro}
	If there exists an algorithm which solves $\SUMInd$ with
	preprocessing space $S$ and $T$ probes, then
	$ST=\widetilde\Omega(N^2)$.
\end{conjecture}

\begin{conjecture}[\cite{GKLP17}]
	\label{conj:strongintro}
	If there exists an algorithm which solves $\SUMInd$ with
	$T=\widetilde O(N^{1-\delta})$ probes for some $\delta>0$ then
	$S=\widetilde\Omega(N^2)$.
\end{conjecture}

These conjectures are in ascending order of strength:
\begin{center}
Conjecture~\ref{conj:strongintro} $\Rightarrow$ Conjecture~\ref{conj:dvintro}
$\Rightarrow$ Conjecture~\ref{conj:weakintro}.
\end{center}

In terms of lower bounds,
Demaine and Vadhan~\cite{DV01} showed that any $1$-probe
data structure for $\ThreeSUMInd$ requires space $S = \Omegat(N^2)$. They leave the case
of $T>1$ open. Goldstein \emph{et al.}~\cite{GKLP17} established connections between Conjectures~\ref{conj:weakintro} and~\ref{conj:strongintro} and the hardness of Set Disjointness, Set Intersection, Histogram Indexing and Forbidden Pattern Document Retrieval.

\subsection{Our Results}\label{sec:intro-results}

Our contributions are three-fold. First, we show better algorithms for
$\ThreeSUMInd$, refuting Conjecture~\ref{conj:strongintro}. Our construction
relies on combining the classical Fiat--Naor inversion algorithm, originally designed for cryptographic applications, with hashing. 
Secondly, we improve the lower bound of \cite{DV01} to arbitrary $T$.
Moreover, we generalize this lower bound to a range of geometric problems, such as
3 points on a line, polygon containment, and others.
As we argue later, any asymptotic improvement to our lower bound will result in a major breakthrough in static data structure lower bounds.

Finally, we show how to use the conjectured hardness of $\ThreeSUMInd$ for a new cryptographic application: namely, designing cryptographic functions that remain secure with massive amounts of preprocessing.
We show how to construct
one-way functions in this model assuming the hardness of a natural average-case variant of $\ThreeSUMInd$.
Furthermore, we prove that this construction generalizes to an \emph{explicit equivalence}
between certain types of hard data structure problems and OWFs in this preprocessing model.
This setting can also be interpreted as security against backdoored random oracles, a problem of grave concern in the modern world.

We describe these results in more detail below.

\subsubsection{Upper bound for $\ThreeSUMInd$}

\begin{theorem}\label{thm:intro-ub}
For every $0\leq\delta\leq 1$, there is an \emph{adaptive} data structure for \ThreeSUMInd with space $S=\widetilde{O}(N^{2-\delta})$ and query time $T=\widetilde{O}(N^{3\delta})$.
\end{theorem}
In particular, Theorem~\ref{thm:intro-ub} implies that by setting $\delta=0.1$, we get a data structure that solves $\ThreeSUMInd$ in space $S=\widetilde{O}(N^{1.9})$ and $T=\widetilde{O}(N^{0.3})$ probes, and thus refutes Conjecture~\ref{conj:strongintro}.

In a nutshell, the upper bound starts by considering the function $f(i,j)
= a_i+a_j$. This function has a domain of size $N^2$ but a potentially much
larger range. In a preprocessing step, we design a hashing procedure to convert
$f$ into a function $g$ with a range of size $O(N^2)$ as well, such that
inverting $g$ lets us invert $f$.  Once we have such a function, we use
Fiat and Naor~\cite{fn00}'s general space-time tradeoff
for inverting functions, which gives non-trivial data structures for
function \emph{inversion} as long as function \emph{evaluation} can be done
efficiently. Due to our definitions of the functions $f$ and $g$, we can
efficiently compute them at every input, which leads to efficient inversion of
$f$, and, therefore, an efficient solution to $\ThreeSUMInd$. For more details,
see Section~\ref{sec:ub}. We note that prior to this work, the result of Fiat
and Naor~\cite{fn00} was recently used by Corrigan-Gibbs and Kogan~\cite{cgk18}
for other algorithmic and complexity applications. In a concurrent work,
Kopelowitz and Porat~\cite{kp19} obtain a similar upper bound for
$\ThreeSUMInd$.

\subsubsection{Lower bound for $\ThreeSUMInd$ and beyond}

We show that any algorithm for $\SUMInd$ that uses a small number of probes requires large space, as expressed formally in Theorem~\ref{thm:lb-3sum}.
\begin{theorem}\label{thm:lb-3sum}
	For every \emph{non-adaptive} algorithm that uses space $S$ and query time $T$
	and solves $\SUMInd$, it holds that $S=\widetilde{\Omega}(N^{1+1/T})$.
\end{theorem}

The lower bound gives us meaningful (super-linear) space bounds for nearly
logarithmic $T$. Showing super-linear space bounds for static data structures
for $T = \omega(\log N)$ probes is a major open question with significant
implications~\cite{Siegel04, patrascu08structures, PTW10, Lar12,dgw19}.

The standard way to prove super-linear space lower bounds for $T=O(\log N)$ is
the so-called \emph{cell-sampling} technique. Applying this technique amounts
to showing that one can recover a fraction of the input by storing a subset of
data structure cells and then using an incompressibility argument.  This
technique applies to data structure problems which have the property that one
can recover some fraction of the input given the answers to \emph{any}
sufficiently large subset of queries.

Unfortunately, the $\ThreeSUMInd$ problem does not have this property and the
cell-sampling technique does not readily apply. Instead we use a different
incompressibility argument, closer to the one introduced by Gennaro and
Trevisan in~\cite{GT00} and later developed in~\cite{DTT10,dod17}. We argue
that given a sufficiently large random subset of cells, with high probability
over a random choice of input, it is possible to recover a constant fraction of
the input. It is crucial for our proof that the input is chosen at random after
the subset of data structure cells, yielding a lower bound only for
\emph{non-adaptive} algorithms.

Next, we show how to extend our lower bound to other data structure problems.
For this, we define $\ThreeSUMInd$-hardness, the data structure analogue of
$\ThreeSUM$-hardness. In a nutshell, a data structure problem is
$\ThreeSUMInd$-hard if there exists an efficient data structure reduction from
$\ThreeSUMInd$ to it. We then show how to adapt known reductions from
$\ThreeSUM$ to many problems in computational geometry and obtain efficient
reductions from $\ThreeSUMInd$ to their data structure counterparts. This in
turns implies that the lower bound in Theorem~\ref{thm:lb-3sum} carries over to
these problems as well.

\subsubsection{Cryptography against massive preprocessing attacks}
\def\BD{\mathsf{BD}}

In a seminal 1980 paper, Hellman~\cite{Hel80} initiated the study of algorithms for inverting (cryptographic) functions with preprocessing. In particular, given a function $F:[N] \to [N]$ (accessed as an oracle), an adversary can run in unbounded time and produce a data structure of $S$ bits.\footnote{The unbounded preprocessing time is amortized over a large number of function inversions. Furthermore, typically the preprocessing time is $\widetilde{O}(N)$.} Later, given access to this data structure and (a possibly uniformly random) $y\in [N]$ as input, the goal of the adversary is to spend $T$ units of time and invert $y$, namely output an $x \in [N]$ such that $F(x) = y$.  It is easy to see that bijective functions $F$ can be inverted at all points $y$ with space $S$ and time $T$ where $ST=O(N)$. Hellman showed that a random function $F$ can be inverted in space $S$ and time $T$ where $S^2T = O(N^2)$, giving in particular a solution with $S =T=O(N^{2/3})$. Fiat and Naor~\cite{fn00} provided a rigorous analysis of Hellman's tradeoff and additionally showed that a worst-case function can be inverted on a worst-case input in space $S$ and time $T$ where $S^3T = O(N^3)$, giving in particular a solution with $S=T=O(N^{3/4})$.
A series of follow-up works~\cite{BBS06,DTT10,AACKPR17} studied time-space tradeoffs for inverting one-way permutations, one-way functions and pseudorandom generators. In terms of lower bounds, Yao~\cite{Yao90} showed that for random functions (and permutations) $ST=\Omega(N)$.
Sharper lower bounds, which also quantify over the success probability and work for other primitives such as pseudorandom generators and hash functions, are known from recent work~\cite{GGKT05,Unruh07,
dod17,CDGS18,CDG18,AACKPR17}.

Hellman's method and followups have been extensively used in practical cryptanalysis, for example in the form of so-called ``rainbow tables''~\cite{Oechslin}. With the increase in storage and available computing power (especially to large organizations and nation states), even functions that have no inherent weakness could succumb to preprocessing attacks.
In particular, when massive amounts of (possibly distributed) storage is at the adversary's disposal, $S$ could be $\Omega(N)$, and the preprocessed string could simply be the function table of the inverse function $F^{-1}$ which allows the adversary to invert $F$ by making a
single access to the $S$ bits of preprocessed string.

One way out of this scenario is
to re-design a new function $F$ with a larger domain.  This is a time-consuming and complex process~\cite{aes,sha}, taking several years, and is fraught with the danger that the new function, if it does not undergo sufficient cryptanalysis, has inherent weaknesses, taking us out of the frying pan and into the fire.

We consider an alternative method that {\em immunizes} the function $F$ against large amounts of preprocessing. In particular, we consider an adversary that can utilize $S \gg N$ bits of preprocessed advice, but can only access this advice by making a limited number of queries, in particular $T \ll N$. This restriction is reasonable when accessing the adversary's storage is expensive, for example when the storage consists of slow but massive memory, or when the storage is distributed across the internet, or when the adversary fields a stream of inversion requests. (We note that while we restrict the number of queries, we do not place any restrictions on the runtime.)

In particular, we seek to design an {\em immunizing compiler} that uses oracle access to $F$ to compute a function $G(x) = G^F(x)$. We want $G$ to remain
secure (for example, hard to invert)
	even against an adversary that
	can make $T$ queries to a preprocessed string of length $S$ bits. Both the preprocessing and the queries can depend on the design of the compiler $G$.
Let $G: [N'] \to [N']$.
	To prevent the inverse table
	attack (mentioned above), we require that $N' > S$.

\paragraph{From Data Structure Lower Bounds to Immunizing Compilers.}
We show how to use data structure lower bounds to construct immunizing compilers.
We illustrate such a compiler here assuming the hardness of the $\ThreeSUMInd$
problem. The compiler proceeds in two steps.
\begin{enumerate}
	\item First, given oracle access to a random
function $F:[2N]\to[2N]$, construct a new (random)
function $F':[N]\to[N^2]$
by letting $F'(x) = F(0,x) \| F(1,x)$.
\item Second, let $G^F(x,y) = F'(x) + F'(y)$ (where the addition is interpreted, e.g., over the integers).
	\end{enumerate}

Assuming the hardness of $\ThreeSUMInd$ for space $S$ and $T$ queries,
we show that this construction is one-way against adversaries with $S$ bits of
preprocessed advice and $T$ online queries. (As stated before, our result is actually stronger: the function remains uninvertible even if the adversary could run for unbounded time in the online phase, as long as it can make only $T$ queries.)
Conjecture~\ref{conj:dvintro}
of Demaine and Vadhan, for example, tells us that
this function is uninvertible as long as $ST = N^{2-\epsilon}$ for any constant $\epsilon>0$. In other words, assuming (the average-case version of) the $\ThreeSUMInd$ conjecture of \cite{DV01}, this function is as uninvertible as a {\em random} function with the same domain and range.

This highlights another advantage of the immunization approach: assume that we have several functions (modeled as independent random oracles) $F_1,F_2,\ldots,F_\ell$ all of which are about to be obsolete because of the increase in the adversary's space resources. Instead of designing $\ell$ independent new functions $F_1',\ldots,F_\ell'$, one could use our {\em immunizer} $G$ to obtain, in one shot, $F_i' = G^{F_i}$ that are as uninvertible as $\ell$ new random functions.

\paragraph{A General Connection.}
In fact, we show a much more general connection between (average-case) data structure lower bounds and immunizing compilers. In more detail, we formalize a data structure problem by a function $g$ that takes as input the data $d$ and a ``target'' $y$ and outputs a ``solution'' $q$.
In the case of $\ThreeSUMInd$, $d$ is the array of $n$ numbers $a_1,\ldots,a_n$, and $q$ is
a pair of indices $i$ and $j$ such that $a_i+a_j=y$. We identify a key property of the data structure problem, namely {\em efficient query generation}. The data structure problem has an efficient query generator if there is a function that, given $i$ and $j$, makes a few queries to $d$ and outputs $y$ such that $g(d,y) = (i,j)$. In the case of $\ThreeSUMInd$, this is just the function that looks up $a_i$ and $a_j$ and outputs their sum.

We then show that any (appropriately hard) data structure problem with an efficient query generator gives us a one-way function in the preprocessing model. In fact, in Section~\ref{sec:cryptods}, we show an equivalence between the two problems.

\paragraph{The Necessity of Unproven Assumptions.}
The one-wayness of our compiled functions rely on an unproven assumption, namely the hardness of the $\ThreeSUMInd$ problem with relatively large space and time (or more generally, the hardness of a data structure problem with an efficient query generator).
We show that unconditional constructions are likely hard to come by in that they would result in significant implications in circuit complexity.

In particular,
a long-standing open problem in computational complexity is to find a function $f : \{0, 1\}^n \to \{0, 1\}^n$ which cannot be computed by binary circuits of linear size $O(n)$ and depth $O(\log n)$~\cite[Frontier~3]{Valiant, AB2009}. We show that even a weak one-way function in the random oracle model with preprocessing (for specific settings of parameters) implies a super-linear circuit lower bound. Our proof in Section~\ref{sec:cktlb} employs the approach used in several recent works~\cite{dgw19,viola2018lower,cgk18,ramamoorthy2019equivalence}.

\paragraph{Relation to Immunizing Against Cryptographic Backdoors.}
Backdoors in cryptographic algorithms pose a grave concern~\cite{DBLP:conf/ccs/CheckowayMGFC0H16, DBLP:conf/uss/CheckowayNE0LRBMSF14,blogpost}, and a natural question is whether one can modify an entropic but imperfect (unkeyed) function, which a powerful adversary may have tampered with, into a function which is provably hard to invert even to such an adversary. In other words, can we use a ``backdoored'' random oracle to build secure cryptography? One possible formalization of a backdoor is one where an unbounded offline adversary may arbitrarily preprocess the random oracle into an exponentially large lookup table to which the (polynomial-time) online adversary has oracle access. It is easy to see that
this formalization is simply an alternative interpretation of (massive) preprocessing attacks.
Thus, our result shows how to construct one-way functions in this model assuming the hardness of a natural average-case variant of $\ThreeSUMInd$.

On immunizing against backdoors, a series of recent works~\cite{DGGJR15,BFM18,RTYZ18,FJM18} studied backdoored primitives including pseudorandom generators and hash functions.   In this setting, the attacker might be given some space-bounded backdoor related to a primitive, which could allow him to break the system more easily.
In particular, backdoored hash functions and random oracles are studied in~\cite{BFM18,FJM18}. Both of them observe that immunizing against a backdoor
for a single unkeyed hash function might be hard.  For this reason,
\cite{BFM18} considers the problem of combining two random oracles (with two
independent backdoors).
Instead, we look at the case of a single random oracle
but add a restriction on the size of the advice.
\cite{FJM18} considers the setting of keyed functions such as (weak) pseudorandom functions, which are easier to immunize than unkeyed functions of the type we consider in this work.

\paragraph{The BRO model and an alternative immunization strategy.}
As mentioned just above, the recent work of~\cite{BFM18} circumvents the problem of massive preprocessing
in a different way, by assuming the existence of at least two independent (backdoored)
random oracles. This allows them to use techniques from two-source extraction
and communication complexity to come up with an (unconditionally secure) immunization strategy. A feature of their approach is that they can tolerate unbounded preprocessing that is {\em separately} performed on the two (independent) random oracles.

\paragraph{Domain extension and indifferentiability.}
Our immunization algorithm is effectively a domain extender for the function
(random oracle) $F$. While it is too much to hope that $G^F$ is
indifferentiable from a random oracle~\cite{DGHM13}, we show that it could
still have interesting cryptographic properties such as one-wayness. We leave
it as an interesting open question to show that our compiler preserves other
cryptographic properties such as pseudorandomness, or alternatively, to come up
with other compilers that preserve such properties.

\subsection{Other related work}

\paragraph{Non-uniform security, leakage, and immunizing backdoors.}
A range of work on non-uniform security, preprocessing attacks,
leakage, and immunizing backdoors can all be seen as addressing
the common goal of achieving security against
powerful adversaries that attack a cryptographic primitive given
access to some ``advice'' (or ``leakage'' or ``backdoor information'')
that was computed in advance during an unbounded preprocessing phase.

On non-uniform security of hash functions, recent works \cite{Unruh07,
dod17,CDGS18} studied the \emph{auxiliary-input random-oracle model} in which an
attacker can compute arbitrary $S$ bits of leakage before attacking the system
and make $T$ additional queries to the random oracle.  Although our model is similar in
that it allows preprocessed leakage of a random oracle, we differ significantly
in two ways: the size of the leakage is larger, and the attacker only has
oracle access to the leakage. Specifically, their results and technical tools only apply to the setting where
the leakage is smaller than the random oracle truth table, whereas our
model deals with larger leakage.  Furthermore, the random
oracle model with auxiliary input allows the online adversary to access and
depend on the leakage arbitrarily while our model only allows a bounded
number of oracle queries to the leakage; our model is more realistic
for online adversaries with bounded time and which cannot read the entire
leakage at query time.

\paragraph{Kleptography.}
The study of backdoored primitives is also related to --- and sometimes falls within the field of --- \emph{kleptography},  originally introduced by Young and Yung \cite{YY97,YY96b,YY96a}.
A \emph{kleptographic attack} ``uses cryptography against
cryptography'' \cite{YY97}, by
changing the behavior of a cryptographic system in a fashion undetectable to an honest user with black-box access to the cryptosystem, such that the use of the modified system leaks some secret information (\emph{e.g.}, plaintexts or key material)
to the attacker who performed the modification.
An example of such an attack might be to modify
the key generation algorithm of an encryption scheme such that an adversary in possession of a
``backdoor'' can derive the private key from the public key, yet an honest user finds the
generated key pairs to be indistinguishable from correctly produced ones.

\paragraph{Data-structure versions of problems in fine-grained complexity.}
While the standard conjectures about the hardness of CNF-SAT, \SUM, OV and APSP concern \emph{algorithms}, the OMV conjecture claims a \emph{data structure} lower bound for the Matrix-Vector Multiplication problem. While algorithmic conjectures help to understand \emph{time} complexity of the problems, it is also natural to consider data structure analogues of the fine-grained conjectures in order to understand \emph{space} complexity of the corresponding problems. Recently, Goldstein \emph{et al.}~\cite{GKLP17,glp17} proposed data structure variants of many classical hardness assumptions (including \SUM and OV). Other data structure variants of the \SUM problem have also been studied in~\cite{DV01,bw09,cl15,ccilmo19}. In particular, Chan and Lewenstein~\cite{cl15} use techniques from additive combinatorics to give efficient data structures for solving \SUM on \emph{subsets} of the preprocessed sets.

\section{Preliminaries}

\subsection{Notation}

When an uppercase letter represents an integer, we use the convention that the
associated lowercase letter represents its base-$2$ logarithm: $N = 2^n,
S=2^s$, etc. $[N]$ denotes the set $\{1,\dots, N\}$ that we identify with
$\zo^n$. $x\|y$ denotes the concatenation of bit strings $x$ and $y$.
PPT stands for probabilistic polynomial time.

We do not consistently distinguish between random variables and their
realizations, but when the distinction is necessary or useful for clarity, we
denote random variables in boldface.

\subsection{\kSUMInd}

This paper focuses on the variant of \ThreeSUM known as \ThreeSUMInd,
formally defined in \cite{GKLP17}, which can be thought of as a preprocessing
or data structure variant of \ThreeSUM. In fact, all our results extend to the
more general \kSUM and \kSUMInd problems which consider $(k-1)$-wise sums
instead of pairwise sums. We also generalize the definition of \cite{GKLP17} by
allowing the input to be elements of an arbitrary abelian\footnote{This is for
	convenience and because our applications only involve abelian groups; our
results and techniques easily generalize to the non-abelian case.} group. We
use $+$  to denote the group operation.

\begin{definition}\label{def:3SUMInd}
	The problem $\kSUMInd(G, N)$, parametrized by an integer $N\geq k-1$ and an
	abelian group $G$, is defined to be solved by a two-part algorithm
	$\cA=(\cA_1,\cA_2)$ as follows.
	\begin{itemize}
		\item \textbf{Preprocessing phase.}
			$\cA_1$ receives as input a tuple $A=(a_1, \dots, a_N)$ of $N$
			elements from $G$ and outputs a data structure $D_A$ of
			size\footnote{The model of computation in this paper is the word
				RAM model where we assume that the word length is $\Theta(\log
				N)$.  Furthermore we assume that words are large enough to
				contain description of elements of $G$, \emph{i.e.}, $|G|\leq
				N^c$ for some $c>0$. The size of a data structure is the number
				of words (or cells) it contains.\label{fn:elem-size}} 
				at most $S$.  $\cA_1$ is computationally unbounded.
			\item \textbf{Query phase.} Denote by $Z$ the set of
				$(k-1)$-wise sums of elements from $A$: $Z=\{\sum_{i\in
				I} a_i: I\subseteq [N] \wedge |I|=k-1\}$. Given an arbitrary
				query $b\in Z$, $\cA_2$ makes at most $T$ oracle queries
				to $D_A$ and must output $I\subseteq [N]$
				with $|I|=k-1$ such
				that $\sum_{i\in I} a_{i} = b$.\footnote{\label{fn:gen}Without loss of
					generality, we can assume that $D_A$ contains a copy of $A$
					and in this case $\cA_2$ could return the tuple
				$(a_{i})_{i\in I}$ at the cost of $(k-1)$ additional queries.}
	\end{itemize}
	We say that $\cA$ is an $(S, T)$ algorithm for $\kSUMInd(G, N)$.
	Furthermore, we say that $\cA$ is \emph{non-adaptive} if the $T$ queries
	made by $\cA_2$ are non-adaptive 
	(\emph{i.e.}, the indices of the queried cells are only a function of
	$b$).
\end{definition}

\begin{remark}\label{rmk:domain} 
	An alternative definition would have the query $b$ be an arbitrary element
	of $G$ (instead of being restricted to $Z$) and $\cA_2$ return the special
	symbol $\bot$ when $b\in G\setminus Z$. Any algorithm conforming to
	Definition~\ref{def:3SUMInd} --- with undefined behavior for $b\in G\setminus
	Z$ --- can be turned into an algorithm for this seemingly more general
	problem at the cost of $(k-1)$ extra queries: given output $I\subseteq [N]$
	on query $b$, return $I$ if $\sum_{i\in I} a_{i}=b$ and return $\bot$
	otherwise.
\end{remark}

\begin{remark}
	The fact that $\kSUMInd$ is defined in terms of $(k-1)$-wise sums of
	\emph{distinct} elements from $G$ is without
	loss of generality for integers, but prevents the occurrence of
	degenerate cases in some groups. For example, consider the case of
	$\ThreeSUMInd$ for a group $G$ such that
	all elements are of order $2$ (\emph{e.g.}, $(\ZZ/2\ZZ)^{cn}$) then
	finding $(i_1,i_2)$ such that $a_{i_1}+a_{i_2}=0$ has the trivial solution
	$(i,i)$ for any $i\in[N]$.
\end{remark}

\begin{remark}\label{rmk:group-index}
	In order to preprocess the elements of some group $G$, 
	we assume an efficient way to enumerate its elements. 
	More specifically, we assume a time- and space-efficient 
	algorithm for evaluating an injective function $\Index\colon G\to [N^c]$ 
	for a constant $c$. For simplicity, we also assume 
	that the word length is at least $c\log{N}$ so that we can store 
	$\Index(g)$ for every $g\in G$ in a memory cell. For example, 
	for the standard $\ThreeSUMInd$ problem over the integers 
	from $0$ to $N^c$, one can consider the group 
	$G=(\ZZ/m\ZZ, +)$ for $m=2N^c+1$, and the trivial function 
	$\Index(a+m\ZZ)=a$ for $0\leq a<m$. For ease of exposition,
	we abuse notation and write $g$ instead of $\Index(g)$ for
	an element of the group $g\in G$. For example, $g\mod p$ for an integer $p$
	will always mean $\Index(g) \mod p$. 
\end{remark}

The standard $\ThreeSUMInd$ problem (formally introduced in \cite{GKLP17})
corresponds to the case where $G=(\ZZ, +)$ and $k=3$.  In fact, it is usually
assumed that the integers are upper-bounded by some polynomial in $N$, which is
easily shown to be equivalent to the case where $G=(\ZZ/N^c\ZZ, +)$ for some
$c>0$, and is sometimes referred to as \emph{modular} $\ThreeSUM$ when there is no
preprocessing.

Another important special case is when $G=\big((\ZZ/2\ZZ)^{cn}, +\big)$ for
some $c>0$ and $k=3$. In this case, $G$ can be thought of as the group of
binary strings of length $cn$ where the group operation is the bitwise XOR
(exclusive or).  This problem is usually referred to as $\ThreeXOR$ when there
is no preprocessing, and we refer to its preprocessing variant as
$\ThreeXORInd$. In \cite{JV16}, the authors provide some evidence that the
hardnesses of $\ThreeXOR$ and $\ThreeSUM$ are related and conjecture that
Conjecture~\ref{conj:3sum} generalizes to $\ThreeXOR$. We similarly conjecture
that in the presence of preprocessing, Conjecture~\ref{conj:dvintro}
generalizes to $\ThreeXORInd$.

Following Definition~\ref{def:3SUMInd}, the results and techniques in this
paper hold for arbitrary abelian groups and thus provide a unified treatment of
the $\ThreeSUMInd$ and $\ThreeXORInd$ problems. It is an interesting open
question for future research to better understand the influence of the group
$G$ on the hardness of the problem.

\begin{openq}
For which groups is $\kSUMInd$ significantly easier to solve, 
and for which groups does Conjecture~\ref{conj:dvintro} not hold?
\end{openq}

\subsubsection{Average-case hardness}

This paper moreover introduces a new \emph{average-case} variant of \kSUMInd 
(Definition~\ref{def:avg-3SUMInd} below) that,
to the authors' knowledge, has not been stated in prior literature.
Definition~\ref{def:avg-3SUMInd} includes an error parameter $\eps$,
as for the cryptographic applications it is useful to consider
solvers for average-case \kSUMInd that only output correct answers
with probability $\eps<1$.

\begin{definition}\label{def:avg-3SUMInd}
	The \emph{average-case} $\kSUMInd(G,N)$ problem, parametrized by an
	abelian group $G$ and integer $N\geq k-1$, is defined to be solved by a two-part
	algorithm $\cA=(\cA_1,\cA_2)$ as follows.
	\begin{itemize}
		\item \textbf{Preprocessing phase.}
			Let $A$ be a tuple of $N$ elements from $G$ drawn uniformly at
			random and with replacement\footnote{We remark that for the classical version of $\kSUM$, the uniform random distribution of the inputs is believed to be the hardest (see, \emph{e.g.}, \cite{avcase}).}. $\cA_1(A)$ outputs a data structure
			$D_A$ of size at most $S$.  $\cA_1$ has unbounded computational
			power.
		\item \textbf{Query phase.} Given a query $b$ drawn uniformly at random
			in $Z=\{\sum_{i\in I} a_i: I\subseteq [N]\wedge |I|=k-1\}$, and
			given up to $T$ oracle queries to $D_A$, $\cA_2(b)$ outputs
			$I\subseteq [N]$ with $|I|=k-1$ such that $\sum_{i\in I}
			a_{i} = b$.
	\end{itemize}

	We say that $\cA=(\cA_1, \cA_2)$ is an $(S,T,\eps)$ solver for
	$\kSUMInd$ if it answers the query correctly with probability $\eps$
	over the randomness of $\cA$, $A$, and the random query $b$.  When
	$\eps=1$, we leave it implicit and write simply $(S,T)$.
\end{definition}

\begin{remark}
In the query phase of Definition~\ref{def:avg-3SUMInd}, the query $b$
is chosen uniformly at random in $Z$ and not in $G$. As observed in
Remark~\ref{rmk:domain}, this is without loss of
generality for $\eps=1$. When $\eps<1$, the meaningful way to measure $\cA$'s
success probability is as in Definition~\ref{def:avg-3SUMInd}, since
otherwise, if $Z$ had negligible density in $G$, $\cA$ could succeed with
overwhelming probability by always outputting $\bot$.
\end{remark}

\section{Upper bound}\label{sec:ub}
We will use the following data structure first suggested by Hellman~\cite{Hel80} and then rigorously studied by Fiat and Naor~\cite{fn00}.
\begin{theorem}[\cite{fn00}]
\label{thm:fn}
For any function $F\colon \X\to \X$, and for any choice of values $S$ and $T$ such that $S^3 T\geq |\X|^3$, there is a deterministic data structure with space $\widetilde{O}(S)$ which allows inverting $F$ at every point making $\widetilde{O}(T)$ queries to the memory cells and evaluations of $F$.~\footnote{While the result in Theorem 1.1 in~\cite{fn00} is stated for a randomized preprocessing procedure, we remark that a less efficient \emph{deterministic} procedure which brute forces the probability space can be used instead.}
\end{theorem}

We demonstrate the idea of our upper bound for the case of $\ThreeSUM$. Since we are only interested in the pairwise sums of the $N$ input elements $a_1,\ldots,a_N\in G$, we can hash down their sums to a set of size $O(N^2)$. Now we define the function $f(i, j)=a_i+a_j$ for $i,j\in[N]$, and note that its domain and range are both of size $O(N^2)$. We apply the generic inversion algorithm of Fiat and Naor to $f$ with $|\X|=O(N^2)$, and obtain a data structure for $\ThreeSUMInd$.

First, in Lemma~\ref{lem:3sumindexing} we give an efficient data structure for
the ``modular'' version of $\kSUMInd(G, N)$ where the input is an integer
$p=\widetilde{O}(N^{k-1})$ and $N$ group elements $a_1,\ldots,a_N\in G$. Given
query $b\in G$ the goal is to find $(i_1,\ldots,i_{k-1})\in[N]^{k-1}$ such that
$a_{i_1}+\dots+a_{i_{k-1}}\equiv b \mod p$.\footnote{Recall from
Remark~\ref{rmk:group-index} that this notation actually means
$\Index(a_{i_1}+\dots+a_{i_{k-1}})\equiv \Index(b) \mod p$.} Then, in
Theorem~\ref{thm:3sumindexing} we reduce the general case of $\kSUMInd(G, N)$
to the modular case.

\begin{lemma}
\label{lem:3sumindexing}
For every integer $k\geq3$, real $0\leq\delta\leq k-2$, and every integer
$p=\widetilde{O}(N^{k-1})$, there is an \emph{adaptive} data structure which
uses space $S=\widetilde{O}(N^{k-1-\delta})$ and query time
$T=\widetilde{O}(N^{3\delta})$ and solves modular $\kSUMInd(G, N)$: for input
$a_1,\ldots,a_N\in G$ and a query $b\in G$, it outputs $(i_1, \ldots,
i_{k-1})\in[N]^{k-1}$ such that $a_{i_1}+\dots+a_{i_{k-1}}\equiv b \mod p$, if
such a tuple exists.
\end{lemma}
\begin{proof}
Let the $N$ input elements be $a_1,\ldots,a_N\in G$. The data structure stores all $a_i$ (this takes only $N$ memory cells) along with the information needed to efficiently invert the function $f\colon [N]^{k-1}\to G$ defined below.
For $(i_1,\ldots,i_{k-1})\in [N]^{k-1}$, let 
\[
f(i_1, \ldots,
i_{k-1})=a_{i_1}+\dots+a_{i_{k-1}} \mod p \ .
\]
 Note that:
\begin{enumerate}
    \item $f$ is easy to compute. Indeed, given the input, one can compute $f$ by looking at only $k-1$ input elements.
    \item The domain of $f$ is of size $N^{k-1}$, and the range of $f$ is of size $p=\widetilde{O}(N^{k-1})$.
	\item Inverting $f$ at a point $b\in G$ finds a tuple
		$(i_1,\ldots,i_{k-1})\in [N]^{k-1}$ such that $a_{i_1}+\dots+a_{i_{k-1}} \equiv b \mod p$, which essentially solves the modular $\kSUMInd(G, N)$ problem.
\end{enumerate}
Now we use the data structure from Theorem~\ref{thm:fn} with $|\X|=\widetilde{O}(N^{k-1})$ to invert $f$. This gives us a data structure with space $\widetilde{O}(S+N)=\widetilde{O}(S)$ and query time $\widetilde{O}(T)$ for every $S^3T\geq |\X|^3=\widetilde{O}(N^{3(k-1)})$, which finishes the proof.
\end{proof}
It remains to show that the input of $\kSUMInd$ can always be hashed to a set of integers $[p]$ for some $p=\widetilde{O}(N^{k-1})$. While many standard hashing functions will work here, we remark that it is important for our application that the hash function of choice has a time- and space-efficient implementation (for example, the data structure in~\cite{fn00} requires non-trivial implementations of hash functions). Below, we present a simple hashing procedure which suffices for $\kSUMInd$; a more general reduction can be found in Lemma~17 in~\cite{cgk18}.
\begin{theorem}
\label{thm:3sumindexing}
For every integer $k\geq3$ and real $0\leq\delta\leq k-2$, there is an \emph{adaptive} data structure for $\kSUMInd(G,N)$ with space $S=\widetilde{O}(N^{k-1-\delta})$ and query time $T=\widetilde{O}(N^{3\delta})$.
\end{theorem}
In particular, by taking $k=3$ and $\delta=0.1$, we get a data structure which solves \ThreeSUMInd in space $S=\widetilde{O}(N^{1.9})$ and query time $T=\widetilde{O}(N^{0.3})$, and, thus, refutes Conjecture~\ref{conj:strongintro}.
\begin{proof}
Let the $N$ inputs be $a_1,\ldots,a_N\in G$. Let $Z\subseteq [N^c], |Z|<\binom{N}{k-1}$ be the set of $(k-1)$-wise sums of the inputs:
$
Z=\{a_{i_1}+\dots+a_{i_{k-1}}\colon 1\leq i_1< \ldots<i_{k-1}\leq N \}
$.

Let $I=\{N^{k-1},\ldots, 3kcN^{k-1}\log{N}\}$ be an interval of integers. By
the prime number theorem, for large enough $N$, $I$ contains at least
$2cN^{k-1}$ primes. Let us pick $n=\log{N}$ random primes $p_1,\ldots,p_n$ from
$I$. For two distinct numbers $z_1, z_2 \in Z$, we say that they have
a collision modulo $p$ if $z_1\equiv z_2\mod p$.

Let $g\in G$ be a positive query of $\kSUMInd(G,N)$, that is, $b=\Index(g)\in Z$.
First, we show that with high probability (over the choices of $n$ random
primes) there exists an $i\in[n]$ such that for every $z\in Z\setminus\{b\}$, $z\not\equiv b\mod p_i$. Indeed, for every $z\in Z\setminus\{b\}$, we
have that $(z-b)$ has at most $\log_{N^{k-1}}(N^c)=c/(k-1)$ prime factors from
$I$. Since $|Z|<\binom{N}{k-1}$, at most $c\binom{N}{k-1}/(k-1)$ primes from
$I$ divide $(z-b)$ for some $z\in Z$. Therefore, a random prime from $I$
gives a collision between $b$ and some $z\in Z\setminus\{b\}$ with probability at most 
\[
\frac{c\binom{N}{k-1}}{k-1}\cdot\frac{1}{2cN^{k-1}}
\leq 
\frac{cN^{k-1}}{(k-1)(k-1)!}\cdot\frac{1}{2cN^{k-1}}
=
\frac{1}{2(k-1)(k-1)!}\leq \frac{1}{2^k} \ .
\]
Now we have that for every $b\in Z$, the probability that there exists an
$i\in[n]$ such that $b$ does not collide with any $z\in Z\setminus\{b\}$ modulo
$p_i$, is at least $1-(2^{-k})^n=1-N^{-k}$. Therefore, with probability at
least $1-1/N$, a random set of $n$ primes has the following property: for every
$b\in Z$ there exists an $i\in[n]$ such that $b$ does not collide with any
$z\in Z\setminus\{b\}$ modulo $p_i$. Since such a set of $n$ primes exists, the
preprocessing stage of the data structure can find it deterministically.

Now we construct $n=\log{N}$ modular $\kSUMInd(G, N)$ data structures (one for each $p_i$), and separately solve the problem for each of the $n$ primes. This results in a data structure as guaranteed by Lemma~\ref{lem:3sumindexing} with a $\log{N}$ overhead in space and time. The data structure also stores the inputs $a_1,\ldots,a_N$. Once it sees a solution modulo $p_i$, it checks whether it corresponds to a solution to the original problem. Now correctness follows from two observations. Since the data structure checks whether a solution modulo $p_i$ gives a solution to the original problem, the data structure never reports false positives. Second, the above observation that for every $b\in Z$ there is a prime $p_i$ such that $b$ does not collide with other $z\in Z$, a solution modulo $p_i$ will correspond to a solution of the original problem (thus, no false negatives can be reported either).
\end{proof}
\begin{remark}
A few extensions of Theorem~\ref{thm:3sumindexing} are in order.
\begin{enumerate}
\item
The result of Fiat and Naor~\cite{fn00} also gives an efficient \emph{randomized} data structure. Namely, there is a randomized data structure with preprocessing running time $\widetilde{O}(|\X|)$, which allows inverting $F$ at \emph{every} point with probability at least $1-1/|\X|$ over the randomness of the preprocessing stage.
Thus, the preprocessing phase of the randomized version of Theorem~\ref{thm:fn}
runs in quasilinear time $\widetilde{O}(|\X|)=\widetilde{O}(N^{k-1})$ (since
sampling $n=\log{N}$ random primes from a given interval can also be done in
randomized time $\widetilde{O}(1)$). This, in particular, implies that the
preprocessing time of the presented data structure for $\ThreeSUMInd$ is
optimal under the $\ThreeSUM$ Conjecture (Conjecture~\ref{conj:3sum}). Indeed,
if for $k=3$, the preprocessing time was improved to $N^{2-\eps}$, then one
could solve $\ThreeSUM$ by querying the $N$ input numbers in (randomized or expected) time $N^{2-\eps}$.
\item
For the case of random inputs (for example, for inputs sampled
as in Definition~\ref{def:avg-3SUMInd}), one can achieve a better time-space
trade-off. Namely, if the inputs $a_1,\ldots,a_N$ are uniformly random numbers from a range of size at least $\Omega(N^{k-1})$, then for every $0\leq \delta\leq k-2$ there is a data structure with space $S=\widetilde{O}(N^{k-1-\delta})$ and query time $T=\widetilde{O}(N^{2\delta})$ (with high probability over the randomness of the input instances). This is an immediate generalization of Theorem~\ref{thm:3sumindexing} equipped with the analogue of Theorem~\ref{thm:fn} for a function~\cite{fn00} with low collision probability, which achieves the trade-off of $S^2T=|\X|^2$.
\item
For polynomially small $\eps=1/|\X|^{\alpha}$ (for constant $\alpha$), the trade-off between $S$ and $T$ can be further improved for the $\eps$-approximate solution of $\kSUMInd$, using approximate function inversion by De \emph{et al.}~\cite{DTT10}.
\end{enumerate}
\end{remark}

We have shown how to refute the strong $\SUMInd$ conjecture of \cite{GKLP17} using
techniques from space-time tradeoffs for function inversion~\cite{Hel80,fn00},
specifically the general function inversion algorithm of Fiat and
Naor~\cite{fn00}. A natural open question is whether a more specific function
inversion algorithm could be designed.

\begin{openq}
	Can the space-time trade-off achieved in Theorem~\ref{thm:3sumindexing}  be
	improved by exploiting the specific structure of the $\SUMInd$ problem?
\end{openq}

\section{Lower bound}
\label{sec:lb}

We now present our lower bound: we prove a space-time trade-off of
$S=\widetilde\Omega(N^{1+1/T})$ for any non-adaptive $(S, T)$ algorithm.
While it is weaker than Conjecture~\ref{conj:dvintro}, any
improvement on this result would break a long-standing barrier in static data
structure lower bounds: no bounds better than $T\geq
\Omega(\frac{\log{N}}{\log(S/N)})$ are known, even for the non-adaptive cell-probe
and linear models~\cite{Siegel04, patrascu08structures, PTW10, Lar12,dgw19}.

Our main lower bound (Theorem~\ref{thm:t2}) is proven with respect to 
a slight variant on Definition~\ref{def:3SUMInd}
to which our proof techniques more readily lend themselves,
defined next.

\begin{definition}\label{def:3SUMInd-elem}
The \emph{element version of the $\kSUMInd(G, N)$ problem} is exactly like the
$\kSUMInd(G, N)$ problem (Definition~\ref{def:3SUMInd-elem}),
except that a solution is a set $\{a_i\}_{i\in I}$ of group elements
summing to a given input instead of simply a set $I$ of indices.
\end{definition}

\begin{remark}
It follows from the observation in Footnote~\ref{fn:gen} that
a lower bound for $(S,T)$ algorithms for the element version of the $\kSUMInd(G, N)$ problem
implies a lower bound for $(S, T-k+1)$ algorithms for the $\kSUMInd(G, N)$ problem as defined
in Definition~\ref{def:3SUMInd-elem}. 
\end{remark}

\begin{theorem}\label{thm:t2}
	Let $k\geq 3$ and $N$ be integers, and let $G$ be an abelian group with $|G|\geq N^{k-1}$.
	Then any \emph{non-adaptive} $(S, T)$ algorithm for the element version of $\kSUMInd(G, N)$
	satisfies $S=\widetilde{\Omega}(N^{1+1/T})$.
\end{theorem}

Our proof relies on a compressibility argument similar to \cite{GT00,DTT10},
also known as cell-sampling in the data structure literature~\cite{PTW10}.
Roughly speaking, we show that given an $(S, T)$ algorithm $(\cA_1, \cA_2)$, we
can recover a subset of the input $A$ by storing a randomly sampled subset $V$
of the preprocessed data structure  $D_A$ and simulating $\cA_2$ on all
possible queries: the simulation succeeds whenever the queries made by $\cA_2$
fall inside $V$. Thus, by storing $V$ along with the remaining part of the
input, we obtain an encoding of the entire input. This implies that the length
of the encoding must be at least the entropy of a randomly chosen input.

\begin{proof}
  Consider an $(S, T)$ algorithm $\cA = (\cA_1, \cA_2)$ for the element
  version of $\kSUMInd(G, N)$. For conciseness, in the rest of this proof, 
  we write simply $\kSUMInd(G, N)$ to denote the element version of $\kSUMInd(G, N)$.

  We use $\cA$ to design encoding and decoding procedures for
	inputs of $\kSUMInd(G, N)$: we first sample a subset $V$ of the data
	structure cells which allows us to answer many queries, then
	we argue that we can recover a constant fraction of the input from this set,
	which yields a succinct encoding of the input.
	
	\paragraph{Sampling a subset $V$ of cells.}
	For a query $b\in G$, $\Query(b)\subseteq[S]$ denotes the set of probes made by
	$\cA_2$ on input $b$ (with $|\Query(b)|\leq T$, since $\cA_2$ makes at most $T$
	probes to the data structure). Given a subset $V\subseteq [S]$ of cells, we
	denote by $G_V$ the set of queries in $G$ which can be answered by $\cA_2$
	by only making probes within $V$: $G_V = \{b\in G\,:\, \Query(b)\subseteq V\}$.
	Observe that for a uniformly random set $V$ of size $v$:
	\begin{displaymath}
		\mathbb{E}\big[|G_V|\big]=\sum_{b\in G}\Pr[\Query(b)\subseteq V]
		\geq |G|\frac{\binom{S-T}{v-T}}{\binom{S}{v}}
		=|G|\prod_{i=0}^{T-1} \frac{v-i}{S-i}\geq
		|G|\left(\frac{v-T}{S-T}\right)^T\, ,
	\end{displaymath}
	where the last inequality uses that $a/b\geq (a-1)/(b-1)$ for $a\leq b$.
	Hence, there exists a subset $V$ of size $v$, such that:
	\begin{displaymath}
		|G_V| \geq |G|\left(\frac{v-T}{S-T}\right)^T
		\,,
	\end{displaymath}
	and we will henceforth consider such a set $V$. The size $v$ of $V$ will be
	set later so that $|G_V|\geq |G|/N$.

	\paragraph{Using $V$ to recover the input.}	
	Consider some input $A=(a_1,\dots, a_N)$ for $\kSUMInd(G, N)$. We
	say that $i\in[N]$ is \emph{good} if $a_i$ is output by $\cA_2$ given some
	query in $G_V$. Since queries in $G_V$ can be answered by only storing the
	subset of cells of the data structure indexed by $V$, our decoding
	procedure will retrieve from these cells all the good elements from $A$.
	
	For a set of indices $I\subseteq[N]$, let $a_I=\sum_{i\in I} a_i$ be the
	sum of input elements with indices in $I$. Also, for a fixed set $G_V$
	and $i\in[N]$, let $g(i)\in G_V$ by some element from $G_V$ which can be
	written as a $(k-1)$-sum of the inputs including $a_i$. If there is no such
	element in $G_V$, then let $g(i)=\bot$. Formally, \[g(i)=\min\{g\in
	G_V\colon \exists I\subseteq [N]\setminus\{i\}, |I|=k-2\colon  a_i+a_I=g
	\}\] with the convention that if the minimum is taken over an empty set, then
	$g(i)=\bot$.

	Note that $i\in[N]$ is good if:
	\begin{equation}
		\label{eq:good}
		\big( g(i) \neq \bot \big)
		\wedge \big(\forall J\subseteq [N]\setminus\{i\},|J|=k-1, \, a_J\neq g(i)\big)
		\,.
	\end{equation}
	Indeed, observe that:
	\begin{enumerate}
		\item The first part of the conjunction guarantees that
	there exists $b\in G_V$ which can be decomposed as $b=a_i+a_I$ for $I\subseteq[N]\setminus\{i\}$.
\item The second part of the conjunction guarantees that every decomposition
	$b=a_{J}, |J|=k-1$ contains the elements $a_i$.
		\end{enumerate}
	By correctness of $\cA$, $\cA_2$ outputs a decomposition of its
	input as a sum of $(k-1)$ elements in $A$ if one exists. For $i$ as in
	\eqref{eq:good}, every decomposition $b=a_I$ contains the input $a_i$, and, therefore, 
	$\cA_2(a_I) = (a_{i_1}, \ldots, a_{i_{k-1}})$, where $i\in\{i_1,\ldots, i_{k-1}\}$.
	
	We denote by $N_V\subseteq[N]$ the set of good indices, and compute its
	expected size when $A$ is chosen at random according to the
	distribution in Definition~\ref{def:avg-3SUMInd}, \emph{i.e.}, for each
	$i\in[N]$, $a_i$ is chosen independently and uniformly in $G$.
	\begin{equation}
		\label{eq:egood}
		\mathbb{E}\big[|N_V|\big]
		\geq\sum_{i=1}^N\Pr[ g(i) \neq \bot]
		\Pr[\forall J\subseteq [N]\setminus\{i\},|J|=k-1, \, a_J\neq g(i)\mid  g(i) \neq \bot]
	\end{equation}
	Let $L\subseteq[N]\setminus\{i\}$ be a fixed set of indices of size $|L|=k-3$. Then:
	\begin{align*}
		\Pr[g(i) \neq \bot]
		&= \Pr[ \exists I\subseteq [N]\setminus\{i\}, |I|=k-2\colon  a_i+a_I\in G_V] \\
		&= 1-\Pr[\forall I\subseteq [N]\setminus\{i\}, |I|=k-2\colon a_i+a_I\notin G_V]\\
		&= 1-\Pr[\forall I'\subseteq [N]\setminus\{i\}, |I'|=k-3, \forall i'\in[N]\setminus (I' \cup \{i\})\colon a_i+a_{I'}+a_{i'}\notin G_V]\\
		&\geq 1-\Pr[\forall i'\in [N]\setminus (L\cup\{i\})\colon a_i+a_L+a_{i'}\notin G_V]\\
		&\geq 1-\left(1-\frac{|G_V|}{|G|}\right)^{N-(k-2)}\, ,
	\end{align*}
	where the first inequality follows from setting $I'=L$, the second
	inequality holds because for every $i'\in [N]\setminus (L \cup \{i\})$,
	$a_{i'}$ needs to be distinct from the $|G_V|$ elements $-a_i-a_L+g$ for
	$g\in G_V$. Furthermore:
	\begin{align*}
		&\Pr[\forall J\subseteq [N]\setminus\{i\},|J|=k-1, \, a_J\neq g(i)\mid  g(i) \neq \bot]\\
		&\quad\quad=1-\Pr[\exists J\subseteq [N]\setminus\{i\},|J|=k-1, \, a_J= g(i)\mid  g(i) \neq \bot]\\
		&\quad\quad\geq 1- \sum_{\substack{J\subseteq [N]\setminus\{i\}\\
		|J|=k-1}} \Pr[a_J= g(i)\mid  g(i) \neq \bot]\\
		&\quad\quad\geq 1 - \binom{N-1}{k-1}\cdot\frac{1}{|G|}\geq \frac{1}{2} \, ,
	\end{align*}
	where the first inequality uses the union bound and the last inequality
	uses that $|G|\geq N^{k-1}$. Using the previous two derivations in
	\eqref{eq:egood}, we get:
	\begin{equation}
		\label{eq:egood-final}
		\mathbb{E}\big[|N_V|\big]\geq
		\frac{N}{2}\left(1-\left(1-\frac{|G_V|}{|G|}\right)^{N-(k-2)}\right)
		\geq \frac{N}{4} \, ,
	\end{equation}
	where the last inequality uses that $|G_V|\geq |G|/N$ and
	$(1-1/N)^{N-(k-2)}\leq 1/2$ for large enough $N$.

	\paragraph{Encoding and decoding.} It follows from \eqref{eq:egood-final}
	and a simple averaging argument that with probability at least $1/16$ over
	the random choice of $A$, $N_V$ is of size at least $N/5$. We will
	henceforth focus on providing encoding and decoding procedures for such
	inputs $A$. Specifically, consider the following pair of
	encoding/decoding  algorithms for $A$:
		\begin{itemize}
			\item $\Enc(A)$: given input $A=(a_1,\dots, a_N)$.
				\begin{enumerate}
					\item use $\cA_2$ to compute the set $N_V\subseteq[N]$ of
						good indices.
					\item store $\big(\cA_1(A)_j)_{j\in V}$ and
							$(a_i)_{i\notin N_V}$.
				\end{enumerate}
			\item $\Dec\big(\Enc(A)\big)$: for each $b\in G$, simulate
				$\cA_2$ on input $b$:
			\begin{enumerate}
				\item\label{it:1stcase} If $\Query(b)\subseteq V$, use
					$\big(\cA_1(A)_i\big)_{i\in V}$ (which was stored in
					$\Enc(A)$) to simulate $\cA_2$ and get $\cA_2(b)$. By
					definition of $N_V$, when $b$ ranges over the queries such
					that $\Query(b)\subseteq V$, this step recovers $(a_i)_{i\in
					N_V}$.
				\item Then recover $(a_i)_{i\notin N_v}$ directly from
					$\Enc(A)$.
			\end{enumerate}
		\end{itemize}

		Note that the bit length of the encoding is:
		\begin{displaymath}
			|\Enc(A)|\leq v\cdot w + (N-|N_V|)\log |G|\leq v\cdot
			w + \frac{4N}{5}\log|G|
		\end{displaymath}
		where $w$ is the word length and where the second inequality holds because
		we restrict ourselves to inputs $A$ such that $|N_V|\geq N/5$. By
		a standard incompressibility argument (see for example Fact 8.1 in
		\cite{DTT10}), since our encoding and decoding succeeds with
		probability at least $1/16$ over the random choice of $A$, we need to be
		able to encode at least $|G|^N/16$ distinct values, hence:
		\begin{equation}
			\label{eq:lb-proof}
			v\cdot w + \frac{4N}{5}\log |G|\geq N\log |G| + O(1)
		\end{equation}

		Finally, as discussed before, we set $v$ such that $|G_V|/|G|\geq 1/N$.
		For this, by the computation performed at the beginning of this proof,
		it is sufficient to have:
		\begin{displaymath}
			\left(\frac{v-T}{S-T}\right)^T\geq \frac{1}{N}
			\,.
		\end{displaymath}
		Hence, we set $v= T + (S-T)/N^{1/T}$ and since $T\leq N\leq S$
		(otherwise the result is trivial), \eqref{eq:lb-proof} implies:
		\begin{displaymath}
			S = \widetilde\Omega (N^{1+1/T})
			\qedhere
		\end{displaymath}
\end{proof}

\begin{remark}
	$\kSUMInd(\ZZ/N^c\ZZ, N)$ reduces to
	$\kSUMInd$ over the integers, so our lower bound extends to
	$\kSUMInd(\ZZ, N)$, too.  Specifically, the reduction works as follows: we
	choose $\{\bar 0,\dots, \overline{N^c-1}\}$ as the set of representatives
	of $\ZZ/N^c\ZZ$. Given some input $A\subseteq \ZZ/N^c\ZZ$ for
	$\kSUMInd(\ZZ/N^c\ZZ, N)$, we treat it as a list of integers and build
	a data structure using our algorithm for $\kSUMInd(\ZZ, N)$. Now, given
	a query $\bar b\in\ZZ/N^c\ZZ$, we again treat it as an integer and query
	the data structure at $b, b+N^c,\ldots, b+(k-2)N^c$. The correctness of the reduction follows from the observation that $\bar b = \bar
	a_{i_{1}} + \dots+\bar a_{i_{k-1}}$ if and only if
	$a_{i_1}+\dots+a_{i_{k-1}} \in\{b,b+N^c,\ldots,b+(k-2)N^c\}$.
\end{remark}

As we already mentioned, no lower bound better than $T\geq
\Omega(\frac{\log{N}}{\log(S/N)})$ is known even for the non-adaptive cell-probe
and linear models, so Theorem~\ref{thm:t2} matches the best known lower bounds
for static data structures. An ambitious goal for future research would
naturally be to prove Conjecture~\ref{conj:dvintro}. A first step in this
direction would be to extend Theorem~\ref{thm:t2} to adaptive strategies that
may err with some probability.

\begin{openq}
	Must any (possibly adaptive) $(S, T, \eps)$ algorithm for
	$\ThreeSUMInd(G, N)$ require $S=\tilde\Omega(\eps N^{1+1/T})$?
\end{openq}

\subsection{$\ThreeSUMInd$-hardness} 

Gajentaan and Overmars introduced the notion of $\ThreeSUM$-hardness and
showed that a large class of problems in computational geometry were
$\ThreeSUM$-hard \cite{go95}.  Informally, a problem is $\ThreeSUM$-hard if 
$\ThreeSUM$ reduces to it with $o(N^2)$ computational overhead. These
fine-grained reductions have the nice corollary that the
$\ThreeSUM$ conjecture immediately implies a $\widetilde\Omega(N^2)$ lower
bound for all $\ThreeSUM$-hard problems.
In this section, we consider a similar paradigm of efficient
reductions between data structure problems, leading to the following definition
of $\ThreeSUMInd$-hardness.

\begin{definition}
	A (static) \emph{data structure problem} is defined by a function $g:D\times Q\to Y$
	where $D$ is the set of data (the input to the data structure problem), $Q$ is
	the set of queries and $Y$ is the set of answers. (See, e.g.,~\cite{M99}.)
\end{definition}

\begin{definition}[$\ThreeSUMInd$-hardness]\label{def:tsi-hardness}
	A data structure problem $g$ is \emph{$\ThreeSUMInd$-hard} if, given a data
	structure for $g$ using space $S$ and time $T$ on inputs of size $N$, it is
	possible to construct a data structure for $\ThreeSUMInd$ using space
	$O(S)$ and time $T$ on inputs of size $N$.
\end{definition}

As an immediate consequence of Definition~\ref{def:tsi-hardness}, 
we get that all $\ThreeSUMInd$-hard problems admit the lower bound of
Theorem~\ref{thm:t2}: i.e., the same lower bound as $\ThreeSUMInd$.
This is stated concretely in Corollary~\ref{cor:3sum-hard}.

\begin{corollary}
	\label{cor:3sum-hard}
	Let $g$ be a $\ThreeSUMInd$-hard data structure problem. Any \emph{non-adaptive} data 
	structure for $g$ using space $S$ and time $T$ on inputs of size $N$ must
	satisfy $S=\widetilde\Omega(N^{1+1/T})$.
\end{corollary}

We now give two examples\footnote{This is far from exhaustive. All the
problems from \cite{go95,bhp01} which inspired the two examples listed here
similarly admit efficient data structure reductions from $\ThreeSUMInd$.} of how to
adapt known reductions from $\ThreeSUM$ to $\ThreeSUM$-hard problems and obtain
efficient reductions between the analogous data structure problems.
Corollary~\ref{cor:3sum-hard} then implies a lower bound of
$S=\widetilde\Omega(N^{1+1/T})$ for these problems as well, matching the best
known lower bound for static data structure problems.

\newcommand{\ThreePOL}{\textsf{3POL}\xspace}
\newcommand{\ThreePOLInd}{\textsf{3POL-Indexing}\xspace}
\newcommand{\PC}{\textsf{PC}\xspace}
\newcommand{\PCInd}{\textsf{PC-Indexing}\xspace}

\paragraph{3 points on a line ($\ThreePOL$).}

Consider the following data-structure variant of the
$\ThreePOL$ problem, referred to as $\ThreePOLInd$.
The input is a set $\mathcal{X}=\{x_1,\dots, x_N\}$ of $N$ distinct points in
$\RR^2$.  Given a query $q\in\RR^2$, the goal is to find $\{i, j\}\subseteq
[N]$ such that $x_i$, $x_j$, and $q$ are collinear (or report $\bot$ if no such
$\{i,j\}$ exists).

The following observation was made in \cite{go95} and used to reduce
$\ThreeSUM$ to $\ThreePOL$: for distinct reals $a$, $b$ and $c$,
it holds that $a+b+c=0$ iff $(a, a^3)$, $(b, b^3)$, $(c, c^3)$ are collinear.
We obtain an efficient data-structure reduction from $\ThreeSUMInd$ to
$3$\textsf{POL-Indexing} by leveraging the same idea, as follows. Given
input $A=(a_1,\dots, a_N)$ for $\ThreeSUMInd$, construct $\mathcal{X}
= \big\{(a_i, a_i^3)\,:\, i\in [N]\big\}$ and use it as input to a data
structure for $3$\textsf{POL-Indexing}. Then, given a query $b$ for
$\ThreeSUMInd$, construct the query $(-b, -b^3)$ for
$3$\textsf{POL-Indexing}.  Finally, observe that an answer $\{i,j\}$ such
that $(-b, -b^3)$ is collinear with $(a_i, a_i^3)$, $(a_j, a_j^3)$ is also
a correct answer for $\ThreeSUMInd$, by the previous observation. The resulting
data structure for $\ThreeSUMInd$ uses the same space and time as the original
data structure and hence $\ThreePOLInd$ is $\ThreeSUMInd$-hard.

\paragraph{Polygon containment ($\PC$).} The problem and reduction described here are
adapted from \cite{bhp01}. Consider the following data-structure variant of the
polygon containment problem, denoted by \PCInd: the input is
a polygon $P$ in $\RR^2$ with $N$ vertices. The query is a polygon $Q$ with
$O(1)$ vertices and the goal is to find a translation $t\in\RR^2$ such that
$Q+t\subseteq P$.

We now give a reduction from $\ThreeSUMInd$ to \PCInd. Consider
input $A=\{a_1,\dots, a_N\}$ for $\ThreeSUMInd$ and assume without loss of
generality that it is sorted: $a_1<\dots<a_N$.  
Let $0<\eps<1$.
We now define the following
``comb-like'' polygon $P$: start from the base rectangle defined by opposite
corners $(0,0)$ and $(3\bar{a},1)$, where $\bar{a}$ is an upper bound on the
elements of the $\SUM$ input (\emph{i.e.}, $\forall a\in A, 
a<\bar{a}$).\footnote{Our model assumes such a bound is known; see footnote
\ref{fn:elem-size}. The reduction can also be adapted to work even if the upper
bound is not explicitly known.}
For each $i\in[N]$, add two rectangle ``teeth'' 
defined by corners $(a_i, 1)$, $(a_i+\eps, 2)$ and $(3\bar{a}-a_i-\eps, 1)$, 
$(3\bar{a}-a_i, 2)$ respectively.  Note that for each $i\in[N]$ we have one 
tooth with abscissa in $[0,\bar{a}]$ and one tooth with abscissa in 
$[2\bar{a},3\bar{a}]$, and there are no teeth in the interval 
$[\bar{a},2\bar{a}]$. We then give $P$ as input to a data structure for \PCInd.

Consider a query $b$ for $\ThreeSUMInd$. If $b\geq 2\bar{a}$ we can immediately 
answer $\bot$, since a pairwise sum of elements in $A$ is necessarily less than 
$2\bar{a}$. We henceforth assume that $b<2\bar{a}$. Define the comb $Q$ with 
base rectangle defined by corners $(0, 0)$ and $(3\bar{a}-b, 1)$ and with two 
rectangle teeth defined by corners $(0, 1)$, $(\eps, 2)$ and 
$(3\bar{a}-b-\eps, 1)$, $(3\bar{a}-b, 2)$ respectively.
It is easy to see that there exists a translation $t$ such that $Q +t\subseteq
P$ iff it is possible to align the teeth of $Q$ with two teeth of $P$.
Furthermore, the two teeth of $Q$ are at least $\bar{a}$ apart
along the $x$-axis, because $b<2\bar{a}$ by assumption, which implies 
$3\bar{a}-b>\bar{a}$. Hence, the
leftmost tooth of $Q$ needs to be aligned with a tooth of $P$ with abscissa in
$[0, \bar{a}]$, the rightmost tooth of $Q$ needs to be aligned with a tooth of 
$P$ with abscissa in $[2\bar{a}, 3\bar{a}]$, and the distance between the two 
teeth needs to be exactly $3\bar{a}-b$. In other words, there exists a 
translation $t$ such that $Q+t\subseteq P$ iff there exists 
$\{i,j\}\subseteq [N]$ such that $(3\bar{a}-a_j)- a_i = 3\bar{a}-b$, 
\emph{i.e.}, $a_i+a_j = b$. The resulting data structure for
$\ThreeSUMInd$ uses the same space and time as the data structure for
\PCInd. This concludes the proof that \PCInd is
$\ThreeSUMInd$-hard.

\section{Cryptography against massive preprocessing attacks}
\label{sec:crypto}

\subsection{Background on random oracles and preprocessing}

A line of work initiated by Impagliazzo and Rudich \cite{IR89} 
studies the hardness of a random oracle
as a one-way function. In \cite{IR89} it was shown that a random oracle is
an exponentially hard one-way function against uniform adversaries. The case of
non-uniform adversaries was later studied in \cite{Imp96, Zim98}. Specifically
we have the following result.

\begin{proposition}[\cite{Zim98}]
	\label{prop:nu-ro}
	With probability at least $1-\frac{1}{N}$ over the choice of a random
	oracle $R:\zo^n\to\zo^n$, for all oracle circuits $C$ of size at most $T$:
\begin{displaymath}
	\Pr_{x\gets\zo^n}\bigg[
		C^{R}\big(R(x)\big)\in R^{-1}\big(R(x)\big)\bigg]
	\in \widetilde{O}\left(\frac{T^2}{N}\right) \ .
\end{displaymath}
\end{proposition}

In Proposition~\ref{prop:nu-ro}, the choice of the circuit occurs \emph{after}
the random draw of the oracle: in other words, the description of the circuit
can be seen as a non-uniform advice which depends on the random oracle.
Proposition~\ref{prop:u-ro} is a slight generalization where the adversary is
a uniform Turing machine independent of the random oracle, with oracle access to
an advice of length at most $S$ depending on the random oracle. While the two
formulations are equivalent in the regime $S\leq T$, one advantage of this
reformulation is that $S$ can be larger than the running time $T$ of the
adversary.

\begin{proposition}[Implicit in \cite{DTT10}]
	\label{prop:u-ro}
	Let $\cA$ be a uniform oracle Turing machine whose number of oracle queries
	is $T:\zo^n\to\NN$. For all $n\in\NN$ and $S\in\NN$, with probability at least
	$1-\frac{1}{N}$ over the choice of a random oracle $R:\zo^n\to\zo^n$:
\begin{displaymath}
	\forall\, P\in\zo^S,\;
	\Pr_{x\gets\zo^n}\bigg[
		\cA^{R,P}\big(R(x)\big)\in R^{-1}\big(R(x)\big)\bigg]
	\in O\left(\frac{T(S+n)}{N}\right) \ .
\end{displaymath}
\end{proposition}

In Proposition~\ref{prop:u-ro}, the advice $P$ can be thought of as the result
of a preprocessing phase involving the random oracle. Also, no assumption is
made on the computational power of the preprocessing adversary but it is simply
assumed that the length of the advice is bounded.

\begin{remark}
	Propositions~\ref{prop:nu-ro} and \ref{prop:u-ro} assume a deterministic
	adversary. For the regime of $S>T$ (which is the focus of this work), this assumption is without loss of generality since
	a standard averaging argument shows that for a randomized adversary, there
	exists a choice of ``good'' randomness for which the adversary achieves at
	least its expected success probability. This choice of randomness can be
	hard-coded in the non-uniform advice, yielding a deterministic adversary.
\end{remark}

Note, however, that Proposition~\ref{prop:u-ro} provides no guarantee when
$S\geq N$. In fact, in this case, defining $P$ to be any inverse mapping
$R^{-1}$ of $R$ allows an adversary to invert $R$ with probability one by
making a single query to $P$.
So, $R$ itself can no longer be used as a one-way function when $S\geq
N$ --- but one can still hope to use $R$ to define a new function $f^R$ that is
one-way against an adversary with advice of size $S\geq N$. This idea motivates
the following definition.

\begin{definition}
	\label{def:our-model}
Let $R:\zo^n\to\zo^n$ be a random oracle.
A \emph{one-way function in the random oracle model with $S$ preprocessing}
is an efficiently computable oracle function $f^R:\zo^{n'}\to\zo^{m'}$
such that for any two-part adversary
$\cA=(\cA_1,\cA_2)$ satisfying $|\cA_1(\cdot)|\leq S$ and where $\cA_2$ is PPT,
the following probability is negligible in $n$:\footnote{A negligible
function is one that is in $o(n^{-c})$ for all constants $c$.}
\begin{equation}\label{eqn:owf}
	\Pr_{R,x\gets\zo^n}\left[f^R\left(\cA_2^{R,\cA_1(R)}\left(f^R(x)\right)\right)=f^R(x)\right]\
	.
\end{equation}
We say that $f$ is an $(S,T,\eps)$-one-way function if the probability in
\eqref{eqn:owf} is less than $\eps$ and $\cA_2$ makes at most $T$ random
oracle queries.
\end{definition}

The adversary model in Definition~\ref{def:our-model} is very similar to the
$1$-BRO model of \cite{BFM18}, differing only in having a restriction on
the output size of $\cA_1$. As was noted in \cite{BFM18}, without this
restriction (and in fact, as soon as $S\geq 2^{n'}$ by the same argument as
above), no function $f^R$ can achieve the property given in
Definition~\ref{def:our-model}. \cite{BFM18} bypasses this impossibility
by considering the restricted case of two independent oracles with two
independent preprocessed advices (of unrestricted sizes). Our work bypasses
it in a different and incomparable way, by considering the case of a
single random oracle with bounded advice.

\subsection{Constructing one-way functions from $\kSUMInd$}\label{sec:3sum-owf}

Our main candidate construction of a OWF (Construction~\ref{constr:owf})
relies on the hardness of average-case \kSUMInd.
First, we define what hardness means, then give the constructions and proofs.

\begin{definition}
Average-case \kSUMInd is \emph{$(G,N,S,T,\eps)$-hard}
if the success probability\footnote{Over
	the randomness of $\cA$, $A$, and
	the average-case \kSUMInd query. (Recall: $A$ is \kSUMInd's input.)}
of any $(S,T)$ algorithm $\cA=(\cA_1,\cA_2)$
in answering average-case \kSUMInd(G,N) queries
is at most $\eps$.
\end{definition}

\begin{construction}\label{constr:owf}
	For $N\in\NN$, let $(G,+)$ be an abelian group
	and let
	$R:[N]\to G$ be a random oracle. Our candidate OWF construction has two
	components:
	\begin{itemize}
		\item the function $f^R:[N]^{k-1}\to G$ defined by
			$ f^R(x) = \sum_{i=1}^{k-1} R(x_i)$ for  $x\in[N]^{k-1}$; and
		\item the input distribution, uniform over $\{x\in[N]^{k-1}\,:\,
			x_1\neq\dots\neq x_{k-1}\}$.
	\end{itemize}
\end{construction}

\begin{remark}[Approximate sampling]
	We depart from the standard definition of a OWF by using a nonuniform
	input distribution in our candidate construction. This makes it easier to
	relate its security to the hardness of \kSUMInd. As long as the input
	distribution is efficiently samplable, a standard construction can be used
	to transform any OWF with nonuniform input into a OWF which operates on
	uniformly random bit strings. Specifically, one simply defines a new OWF
	equal to the composition of the sampling algorithm and the original OWF,
	(see \cite[Section 2.4.2]{G01}).

	In our case, since $N!/(N-k+1)!$ is not guaranteed to be a power of $2$,
	the input distribution in Construction~\ref{constr:owf} cannot be sampled
	exactly in time polynomial in $\log N$. However, using rejection sampling,
	it is easy to construct a sampler taking as input $O(\lceil\log N\rceil^2)$
	random bits and whose output distribution is $1/N$-close  in statistical
	distance to the input distribution. It is easy to propagate  this
	exponentially\footnote{Recall that $N=2^n$ and that following
	Definition~\ref{def:our-model}, $n$ is the security parameter. Terms like
``exponential'' or ``negligible'' are thus defined with respect to $n$.} small sampling error without affecting the conclusion of
	Theorem~\ref{thm:crypto} below. A similar approximate sampling occurs when
	considering OWFs based on the hardness of number theoretic problems, which
	require sampling integers uniformly in a range whose length is not
	necessarily a power of two.
\end{remark}

\begin{remark}
	Similarly, the random oracle $R$ used in the construction is not a random
	oracle in the traditional sense since its domain and co-domain are not bit
	strings. If $|G|$ and $N$ are powers of two, then $R$ can be implemented
	exactly by a standard random oracle $\zo^{\log N}\to\zo^{\log|G|}$. If not,
	using a random oracle
	$\zo^{\poly(\lceil\log|G|\rceil)}\to\zo^{\poly(\lceil\log|G|\rceil)}$, and
	rejection sampling, it is possible to implement an oracle $R'$ which is
	$1/N$ close to $R$ in statistical distance. We can similarly propagate this
	$1/N$ sampling error without affecting the conclusion of
	Theorem~\ref{thm:crypto}.
\end{remark}

\begin{theorem}\label{thm:crypto}
	Consider a sequence of abelian groups $(G_N)_{N\geq 1}$ such that
	$|G_N|\geq N^{k-1+c}$ for some $c>0$ and all $N\geq k-1$, and a function
	$S:\NN\to\RR$. Assume that for all polynomial $T$ there exists a negligible
	function $\eps$ such that average-case $\kSUMInd$ is $(G_N, N, S(n),
	T(n), \eps(n))$-hard for all $N\geq 1$ (recall that $n=\log N$). Then the
	function $f$ defined in Construction~\ref{constr:owf} is a one-way function
	in the random oracle model with $S$ preprocessing.
\end{theorem}

The function $f^R$ in Construction~\ref{constr:owf} is designed precisely so
that inverting $f^R$ on input $x$ is equivalent to solving \kSUMInd for the
input $A=\big(R(1),\dots, R(N)\big)$ and query $\sum_{i=1}^{k-1} a_{x_i}$.
However, observe that the success probability of a OWF inverter is defined for
a random input distributed as $\sum_{i\in I} a_i$ where $I\subseteq [N]$ is
a uniformly random set of indices of size $k-1$. In contrast, in average-case
\kSUMInd, the query distribution is uniform over $\{ \sum_{i\in I} a_i \,:\,
	I\subseteq N, |I|=k-1\}$.  These two distributions are not identical
	whenever there is a collision: two sets $I$ and $I'$ such that $\sum_{i\in
		I} a_i = \sum_{i\in I'} a_i$. The following two lemmas show that
		whenever $|G|\geq N^{k-1+c}$ for some $c>0$, there are few enough
		collisions that the two distributions are negligibly close in
		statistical distance, which is sufficient to prove
		Theorem~\ref{thm:crypto}.

\begin{lemma}
	\label{lem:dist}
	Let $N\geq k-1$ be an integer and let $G$ be an abelian group with $|G|\geq
	N^{k-1+c}$ for some $c>0$. Let
	$\mathbf{A}=(\mathbf{a}_1,\dots,\mathbf{a}_N)$ be a tuple of $N$ elements
	drawn with replacement from $G$. Define the following two random variables:
	\begin{itemize}
		\item $\mathbf{X}_1= \sum_{i\in \mathbf{I}} \mathbf{a}_i$ where
			$\mathbf{I}\subseteq[N]$ is a uniformly random set of size $k-1$.
		\item $\mathbf{X}_2$: uniformly random over $\{\sum_{i\in I} \mathbf{a}_i\,:\,
			I\subseteq[N], |I|=k-1\}$.
	\end{itemize}
	Then the statistical distance is $\sdist{(\mathbf{A}, \mathbf{X}_1)-(\mathbf{A}, \mathbf{X}_2)}=
	O\big(1/\sqrt{N^c}\big)$.
\end{lemma}

\begin{proof}
	First, by conditioning on the realization of $\mathbf{A}$:
	\begin{equation}
		\label{eq:stat-dist}
		\sdist{(\mathbf{A},\mathbf{X}_1)-(\mathbf{A},\mathbf{X}_2)}
		=\sum_{A\in G^N}\Pr[\mathbf{A}=A]\sdist{\mathbf{X}_{1|A}-\mathbf{X}_{2|A}}
		\,,
	\end{equation}
	where $\mathbf{X}_{i|A}$ denotes the distribution of $\mathbf{X}_{i}$
	conditioned on the event $\mathbf{A}=A$ for $i\in\{1,2\}$.

	We now focus on a single summand from \eqref{eq:stat-dist} corresponding to
	the realization $\mathbf{A}=A$ and define $Z= \{\sum_{i\in I} a_i\,:\, I\subseteq
		[N], |I|=k-1\}$, the set of $(k-1)$-sums and for $g\in G$, $c_g
		= \left|\left\{I\subseteq [N]\,:\, |I|=k-1 \wedge \sum_{i\in I}a_i
		= g\right\}\right|$ is the number of $(k-1)$-sets of indices whose
		corresponding sum equals $g$. Then we have:
		\begin{displaymath}
			\sdist{\mathbf{X}_{1|A} - \mathbf{X}_{2|A}} = \frac{1}{2}\sum_{g\in Z}
			\left|\frac{1}{|Z|} - \frac{c_g}{{N\choose k-1}}\right| \,.
		\end{displaymath}
	Observe that $c_g\geq 1$ whenever $g\in Z$. We now assume that $|Z|\geq
	\frac{1}{2}{N\choose k-1}$ (we will later only use the following derivation
	under this assumption). Splitting the sum on $c_g>1$:
	\begin{displaymath}
		\sdist{\mathbf{X}_{1|A} - \mathbf{X}_{2|A}}
		= \frac{1}{2}\sum_{g\,:\,c_g=1} \left(\frac{1}{|Z|}
		- \frac{1}{{N \choose k-1}}\right)
		+ \frac{1}{2}\sum_{g\,:\, c_g > 1} \left(\frac{c_g}{{N\choose k-1}} - \frac{1}{|Z|}\right)\,,
	\end{displaymath}
	where we used the trivial upper bound $|Z|\leq {N\choose k-1}$ and the
	assumption that $|Z|\geq\frac{1}{2}{N\choose k-1}$ to determine the sign of
	the quantity inside the absolute value. We then write:
	\begin{align*}
		\sdist{\mathbf{X}_{1|A} - \mathbf{X}_{2|A}}
		&= \frac{1}{2}\sum_{g\,:\,c_g=1} \left(\frac{1}{|Z|}
		- \frac{1}{{N\choose k-1}}\right)
		+ \frac{1}{2}\sum_{g\,:\, c_g > 1} \left(\frac{c_g-1}{{N\choose k-1}}
			+\frac{1}{{N\choose k-1}}- \frac{1}{|Z|}\right)\\
		&\leq \frac{1}{2}\sum_{g\,:\,c_g\geq 1} \left(\frac{1}{|Z|}
		- \frac{1}{{N\choose k-1}}\right)
		+ \frac{1}{2}\sum_{g\,:\, c_g> 1} \frac{c_g-1}{{N\choose k-1}}\\
		&= \frac{1}{2}\sum_{g\,:\,c_g\geq 1} \left(\frac{1}{|Z|}
		- \frac{1}{{N\choose k-1}}\right)
		+ \frac{1}{2}\sum_{g\,:\, c_g\geq 1} \frac{c_g-1}{{N\choose k-1}}
		= \left(1-\frac{|Z|}{{N\choose k-1}}\right)\,,
	\end{align*}
	where the inequality uses again that $|Z|\leq {N\choose k-1}$, and the last
	equality uses that $\sum_{g:c_g\geq 1} c_g = {N\choose k-1}$ and that
	$Z=\{g\,:\, c_g\geq 1\}$.

	We now consider some $\delta\leq 1/2$ which will be set at the end of the
	proof and split the sum in \eqref{eq:stat-dist} on $|Z|\leq
	(1-\delta){N\choose k-1}$:
	\begin{align*}
		\sdist{(\mathbf{A}, \mathbf{X}_1)- (\mathbf{A}, \mathbf{X}_2)}
		&\leq \Pr\left[|Z|\leq {N\choose k-1}(1-\delta)\right]
		+ \delta\cdot \Pr\left[|Z|>{N\choose k-1}(1-\delta)\right]\\
		&\leq \Pr\left[|Z|\leq {N\choose k-1}(1-\delta)\right]
		+ \delta
		\,,
	\end{align*}
	where we used the trivial upper bound $\sdist{\mathbf{X}_{1|A}
	- \mathbf{X}_{2|A}}\leq 1$ when $|Z|\leq (1-\delta){N\choose k-1}$ and the
	upper bound $\sdist{\mathbf{X}_{1|A} - \mathbf{X}_{2|A}}< \delta$  when $|Z|>
	(1-\delta){N\choose k-1}$ by the previous derivation.

	We now use Markov's inequality and Lemma~\ref{lem:size} below to upper
	bound the first summand:
	\begin{align*}
		\sdist{(\mathbf{A},\mathbf{X}_1)-(\mathbf{A}, \mathbf{X}_2)}
		&\leq \frac{1}{\delta {N\choose k-1}}\left({N\choose k-1}
		- \mathbb{E}\big[|Z|\big]\right) + \delta\\
		&\leq \frac{1}{\delta |G|}{N \choose k-1} + \delta
		\leq \frac{1}{\delta(k-1)! N^c} + \delta
		\,.
	\end{align*}
	where the last inequality uses that $|G|\geq N^{k-1+c}$ by assumption.
	Finally, we set $\delta = 1/\sqrt{N^{c}}$ to get the desired conclusion.
\end{proof}

\begin{lemma}
	\label{lem:size}
	Let $N\geq k-1$ be an integer and let $G$ be an abelian group of size at
	least $N$. Let $\mathbf{A}=(\mathbf{a}_1,\dots,\mathbf{a}_N)$ be a tuple of
	$N$ elements drawn with replacement from $G$. Define $Z = \big\{ \sum_{i\in
	I} \mathbf{a}_i\,:\, I\subseteq [N]\wedge |I|=k-1\}$ to be the set of
		$(k-1)$-sums of coordinates of $A$, then:
	\begin{displaymath}
		{N\choose k-1} - \mathbb{E}\big[|Z|\big] \leq \frac{1}{|G|}{N\choose k-1}^2\,.
	\end{displaymath}
\end{lemma}

\begin{proof}
	For each $(k-1)$-set of indices $I\subseteq [N]$, we define the random
	variable $X_I$ to be the indicator that the sum $\sum_{i\in I}\mathbf{a}_i$
	collides with $\sum_{i\in I'} \mathbf{a}_i$ for some $(k-1)$-set of indices
	$I'\neq I$:
	\begin{displaymath}
		X_I = \mathbf{1}\left\{\exists I'\subseteq [N]\,:\, |I'|=k-1 \wedge
			I'\neq I \wedge \sum_{i\in I} \mathbf{a}_i = \sum_{i\in I'} \mathbf{a}_i
		\right \}
		\,.
	\end{displaymath}
	Then, using a union bound and since the probability of a collision is
	$1/|G|$:
	\begin{displaymath}
		\mathbb{E}\big[X_I\big] \leq \sum_{I'\neq I}
		\Pr\left[\sum_{i\in I}\mathbf{a}_i = \sum_{i\in I'} \mathbf{a}_i\right] \leq
		\frac{{N\choose k-1}}{|G|}
		\,.
	\end{displaymath}
	On the other hand, there are at least as many elements in $Z$ as $(k-1)$-sets of
	indices $I\subseteq [N]$ which do not collide with any other $(k-1)$-set:
	\begin{displaymath}
		|Z| \geq \sum_{\substack{I\subseteq [N]\\|I|=k-1}} (1-X_I)
		= {N\choose k-1} - \sum_{\substack{I\subseteq [N]\\|I|=k-1}} X_I
		\,.
	\end{displaymath}
	Combining the previous two inequalities concludes the proof.
\end{proof}

We are now ready to prove Theorem~\ref{thm:crypto}.

\begin{proof}[Proof (Theorem~\ref{thm:crypto})]
	Throughout the proof, we fix $N$ and write $G,S,T$ to denote
	$G_N,S(n),T(n)$ respectively, leaving the parameter $n$ implicit.  Suppose,
	for contradiction, that $f$ is not a one-way function in the random oracle
	model with $S$ preprocessing.  Then there exists $\cA=(\cA_1,\cA_2)$ such
	that $|\cA_1(\cdot)|\leq S$ and $\cA_2$ is PPT, which inverts $f$ with
	probability at least $\delta$ for some non-negligible $\delta$:
	\begin{equation}\label{eq:break}
		\Pr_{R,\mathbf{x}}\left[f^R\left(\cA_2^{R,\cA_1(R)}\left(f^R(\mathbf{x})\right)\right)=f^R(\mathbf{x})\right]\geq\delta\, .
\end{equation}
where $R:[N]\to G$ is a random oracle and $\mathbf{x}\in[N]^{k-1}$ is a random input to
$f^R$ distributed as defined in Construction~\ref{constr:owf}. Then, we use
$\cA$ to build an $(S,T)$ solver $\cA'=(\cA'_1,\cA'_2)$ for $\kSUMInd(G,N)$ as
follows.  Given input $A= (a_1,\dots,a_N)$ for $\kSUMInd(G,N)$, $\cA'_1$
defines random oracle $R:[N]\to G$ such that $R(i) = a_i$ for $i\in [N]$ and
outputs $\cA_1(R)$ --- this amounts to interpreting the tuple $A$ as a function
mapping indices to coordinates. $\cA_2'$ is identical to $\cA_2$.  By
construction, whenever $\cA_2$ successfully inverts $f^R$ ({\it i.e.}, outputs
$x\in[N]^{k-1}$ such that $f^R(x)=b$ for input $b$), then the output of
$\cA_2'$ satisfies $\sum_{i=1}^{k-1} a_{x_i}= b$.

It follows from \eqref{eq:break} that $\cA'$ as described thus far solves
average-case $\kSUMInd(G, N)$ with success probability $\delta$ when given
as input a query distributed as $f^R(\mathbf{x})$. By construction, the
distribution of $f^R(\mathbf{x})$ is identical to the distribution of $\sum_{i\in I}
a_i$ for a uniformly random set $I\subseteq [N]$ of size $k-1$, let $X_1$ denote this
distribution. However, average-case $\kSUMInd(G, N)$ is defined with
respect to a distribution of queries which is uniform over $\{\sum_{i\in I}
a_i\,:\, I\subseteq [N]\wedge |I|=k-1\}$, let us denote
this distribution by $X_2$. By Lemma~\ref{lem:dist}, we have that $\sdist{(A,
X_1) - (A, X_2)} = O(1/\sqrt{N^{c}})$, hence $\cA_2$ solves $\kSUMInd(G,
N)$ for the correct query distribution $X_2$ with probability at least $\delta-
O(1/\sqrt{N^c})$ which is non-negligible since $\delta$ is non-negligible.
Denoting by $T$ the running time of $\cA_2$, we just proved that $\cA'$  is an
$(S, T,\delta-O(1/\sqrt{N^c}))$ adversary for average-case $\kSUMInd(G,
N)$, which is a contradiction.
\end{proof}

We conjecture that $\ThreeSUMInd$ is $(G, N, S, T, \eps)$-hard with $\eps
= \frac{ST}{N^2}$ when $G = (\ZZ/N^c\ZZ, +)$ (the standard $\ThreeSUMInd$
problem) and $G = ((\ZZ/2\ZZ)^{cn}, \oplus)$ (the $\ThreeXORInd$ problem) for
$c> 2$. If this conjecture is true, the previous theorem implies the existence of
(exponentially strong) one-way functions in the random oracle model as long the
preprocessing satisfies $S\leq N^{2-\delta}$ for $\delta> 0$. As per the
discussion below Definition~\ref{def:our-model}, Theorem~\ref{thm:crypto} is
vacuous in the regime where $S=\widetilde\Omega(N^2)$.

\subsection{Cryptography with preprocessing and data structures}
\label{sec:cryptods}

In this section we show that the construction in Section~\ref{sec:3sum-owf} is
a specific case of a more general phenomenon.  Specifically,
Theorem~\ref{thm:eq} below states that the existence of one-way functions in
the random oracle model with preprocessing is equivalent to the existence of
a certain class of hard-on-average data structure problems.
The next two definitions formalize the definitions of a data structure problem
and a solver for a data structure problem.

\begin{definition}
	An \emph{$(S,T,\eps)$-solver} for a data structure problem $g:D\times Q\to Y$
	is a two-part algorithm $\cB=(\cB_1,\cB_2)$ such that:
	\begin{itemize}
		\item $\cB_1$ takes as input $d\in D$ and computes a data structure $\phi(d)$
		such that $|\phi(d)|\leq S$; and
		\item $\cB_2$ takes as input query $q\in Q$, makes at most $T$ queries to $\phi(d)$,
		and outputs $y\in Y$.
	\end{itemize}
	We say that a given execution of $\cB$ \emph{succeeds} if $\cB_2$ outputs $y=g(d,q)$.
\end{definition}

Theorem~\ref{thm:eq} considers
a special class of data structure problems for which a query can be efficiently
generated given its answer, as defined next.

\begin{definition}\label{def:eff-query-gen}
	Let $g:D\times Q\to Y$ be a static data structure problem
	and let $h:D\times Y\to Q$.
	Then $h$ is an \emph{efficient query generator for $g$} if
	$h$ is
	computable in time $\poly(\log{|Q|}, \log{|Y|})$
	and
	\begin{equation}\label{eqn:eff-query-gen}
		\forall d\in D,\, y\in Y,\; g\big(d, h(d, y)\big) = y
		\,.
	\end{equation}
	For any $h$ which is an efficient query generator for $g$,
	we say that \emph{$(g,h)$ is $(S, T, \eps)$-hard}
	if for query distribution $q=h(d, y)$
	where $d\in D,y\in Y$ are uniformly random,
	no $(S, T)$-solver succeeds with probability more than $\eps$.\footnote{For simplicity we consider the uniform distributions on $D$ and $Y$, but all definitions and results easily generalize to arbitrary distributions.}
\end{definition}

\begin{remark}
	For the $\ThreeSUMInd$ problem, $h$ is the function that takes $d
	= (a_1,\ldots,a_n)$ and a pair of indices $y=(i,j)$ and outputs $a_i+a_j$.
	Constructing a corresponding function $g$ for this $h$ is equivalent to
	solving the $\ThreeSUMInd$ problem.
\end{remark}
\begin{remark}
	Let $g,h$ be defined as in Definition~\ref{def:eff-query-gen}.
	Then because $g$ is a function and $h$ satisfies~\eqref{eqn:eff-query-gen}, it holds that
	for any given $d\in D$, the function $h(d,\cdot)$ is injective. That is, for any $d\in D, y,y'\in Y$,
	\begin{equation}\label{eqn:uniqueness}
	h(d,y)=h(d,y') \quad\Rightarrow\quad y=y'\ .
	\end{equation}
\end{remark}

\begin{theorem}
\label{thm:eq}
	There exists a $(S, T, \eps)$-hard data structure with efficient query generation iff
	there exists a $(S, T,\eps)$-hard OWF in the random oracle model with preprocessing.

	More specifically, there is an efficient explicit transformation:
	(1) from any $(S, T, \eps)$-hard data structure with efficient query generation
	to a $(S, T,\eps)$-hard OWF in the random oracle model with preprocessing; and
	(2) from any $(S, T,\eps)$-hard OWF in the random oracle model with preprocessing
	to an explicit construction of a $(S, T, \eps)$-hard data structure.
	For the second transformation, the resulting data structure is always in ${\bf QuasiP}$
	(with respect to its input size), and is in fact in ${\bf P}$
	whenever the input/output size of the underlying OWF is linear in the
	input/output size of the random oracle.
\end{theorem}
\begin{proof}
	We show the two implications in turn.\footnote{Throughout this proof, we assume
	the domain and range of the data structure problem and OWF are bitstrings.
	The proof generalizes to arbitrary domains and ranges. 
	}
	\begin{itemize}
		\item \textbf{DS $\Rightarrow$ OWF.}
			Let $g:\zo^{\tilde{N}}\times\zo^{m'}\to\zo^{n'}$ be a data structure problem, and
			let $h:\zo^{\tilde{N}}\times\zo^{n'}\to\zo^{m'}$ be an efficient query generator for $g$
			such that $(g,h)$ is $(S, T, \eps)$-hard.
			Let $R:\zo^n\to\zo^n$ be a random oracle, such that $\tilde{N}=n2^n$.
			We define an oracle function $f^R:\zo^{n'}\to\zo^{m'}$ as follows:
			\[f^R(x)=h(\hat{R},x)\ ,\]
			where $\hat{R}$ denotes the binary representation of $R$.

			$f$ is a $(S, T,\eps)$-hard OWF in the random oracle model with preprocessing,
			because it is efficiently computable and hard to invert, as proven next.
			Since $h$ is efficiently computable, $f$ runs in time $\poly(n',m')$.

			It remains to show that $f$ is $(S,T,\eps)$-hard to invert.
			Suppose, for contradiction, that this is not the case: namely, that there is
			a two-part adversary $\cA=(\cA_1,\cA_2)$ such that
			\begin{equation}\label{eqn:DS-to-OWF}
				\Pr_{x\gets\zo^{n'}}\left[h\left(R,\cA_2^{\cA_1(R)}\left(h(R,x)\right)\right)=h(R,x)\right]>\eps\ ,
			\end{equation}
			and $\cA_1$'s output size is at most $S$, $\cA_2$ makes at most $T$ queries to $\cA_1(R)$,
			and the probability is also over the sampling of the random oracle $R$.

			We use $\cA$ to build $(\cB_1,\cB_2)$, an $(S,T)$-solver for $g$, as follows.
			On input $d\in \zo^{\tilde{N}}$, $\cB_1$ simply outputs $\phi(d)=\cA_1(d)$.
			On input $q\in\zo^{m'}$, $\cB_2$ runs $\cA_2^{\cA_1(R)}(q)$;
			for each query $\zeta$ that $\cA_2$'s makes to $\cA_1(R)$,
			$\cB_2$ simply queries $\phi(d)$ on $\zeta$ and returns the response to $\cA_2$.

			It follows from \eqref{eqn:uniqueness} and \eqref{eqn:DS-to-OWF} that
			\[\Pr_{\substack{d\gets\zo^{\tilde{N}}\\y\gets\zo^{n'}}}\left[\cB_2^{\phi(d)}(h(d,y))=y\right]\geq\eps\ .\]
			This contradicts the $(S,T,\eps)$-hardness of $(g,h)$.
		\item \textbf{OWF $\Rightarrow$ DS.}
			Let $f^R:\zo^{n'}\to\zo^{m'}$ be a
			$(S, T,\eps)$-hard OWF in the random oracle model with preprocessing,
			for a random oracle mapping $n$ bits to $n$ bits.
			We design a data structure problem $g:\zo^{\tilde{N}}\times\zo^{m'}\to\zo^{n'}$
			and an efficient query generator $h$ for $g$
			such that $\tilde{N}=n2^n$ and $(g,h)$ is $(S, T, \eps)$-hard, as follows.
			\begin{itemize}
				\item $h(d,y)=f^d(y)$.
				\item $g(d,q)=\min\{y\in Y ~:~ f^d(y)=q\}$.\footnote{
				For the purpose of this proof,
				$g(d,\cdot)$ can be any inverse of $f^d$ that is computable in time $O(2^{n'})$.
				We use the concrete example of $g(d,q)=\min\{y\in Y ~:~ f^d(y)=q\}$
				for ease of exposition.}
			\end{itemize}
			$h$ is computable in time $\poly(n',m')$, as required by Definition~\ref{def:eff-query-gen},
			because $f^d$ is efficiently computable (in its input size).
			Furthermore, $h$ satisfies \eqref{eqn:eff-query-gen} since $g$ is,
			by construction, an inverse of $h$.

			Next, we show that $(g,h)$ is $(S, T, \eps)$-hard.
			Suppose the contrary, for contradiction.
			Then there exists an $(S,T)$-solver $\cB=(\cB_1,\cB_2)$ for $g$
			that succeeds with probability greater than $\eps$
			on query distribution $q=h(d,y)=f^d(y)$ where $d,y$ are uniformly random.
			Then $\cB$ is quite literally an inverter for the OWF $f$,
			where $d$ corresponds to the random oracle and $q$ corresponds to
			the challenge value to be inverted: by assumption, $\cB$ satisfies
			\[\Pr_{\substack{d\gets(\zo^{n}\to\zo^{n})\\y\gets\zo^{n'}}}\left[f^d\left(\cB_2^{\cB_1(d)}\left(f^d(y)\right)\right)=f^d(y)\right]>\eps\ .\]
			This contradicts the $(S, T,\eps)$-hardness of $f$.

			Finally, $g$ is computable in ${\bf DTIME}[2^{n'}\cdot\poly(n')]$,
			since it can be solved by exhaustively searching all $y\in\zo^{n'}$
			and outputting the first (i.e., minimum) such that $f^d(y)=q$.
			Note that $n',m'\in\poly(n)$ since $n',m'$ are the input and output sizes
			of a OWF with oracle access to a random oracle mapping $n$ bits to $n$ bits.
			Hence, $g$ is computable in time quasipolynomial in $|d|=\tilde{N}=n2^n$,
			i.e., the size of $g$'s first input.
			In particular, $g$ is computable in time $\poly(\tilde{N})$
			whenever $n',m'\in O(n)$.
			\qedhere
	\end{itemize}
\end{proof}

\begin{remark}
As an example, a one-way function  $f^R\colon\zo^{5n}\to\zo^{5n}$ in the random oracle model with preprocessing $S=2^{3n}$ would give an \emph{adaptive} data structure lower bound for a function with $N$ inputs, $N^5$ outputs, space $S=\Omega(N^3/\poly\log(N))$ and query time $T=\poly\log(N)$. Finding such a function is a big open problem in the area of static data structures~\cite{Siegel04, patrascu08structures, PTW10, Lar12,dgw19}.

\end{remark}

\subsubsection{Cryptography with preprocessing and circuit lower bounds}\label{sec:cktlb}
Although the existence of cryptography in the random oracle model with preprocessing does not have such strong implications in complexity as the existence of regular cryptography, in Theorem~\ref{thm:circuits} we show that it still has significant implications in circuit complexity.

A long-standing open problem in computational complexity is to find a function $f\colon\zo^n\to\zo^n$ which cannot be computed by binary circuits of linear size $O(n)$ and logarithmic depth $O(\log{n})$~\cite[Frontier~3]{Valiant, AB2009}.\footnote{The same question is open even for series-parallel circuits~\cite{Valiant}. A circuit is called series-parallel if there exists a numbering $\ell$ of the circuit's nodes s.t. for every wire $(u,v), \ell(u)<\ell(v)$, and no pair of arcs $(u,v), (u',v')$ satisfies $\ell(u)<\ell(u')<\ell(v)<\ell(v')$.}
We now show that a weak one-way function with preprocessing would resolve this question.

First we recall the classical result of Valiant~\cite{Valiant} asserting that every linear-size circuit of logarithmic depth can also be efficiently computed in the common bits model.
\begin{definition}
A function $f=(f_1,\ldots,f_m)\colon\zo^n\to\zo^m$ has an $(s,t)$-solution in the common bits model if there exist $s$ functions $h_1,\ldots,h_s\colon\zo^n\to\zo$, such that each $f_i$ can be computed from $t$ inputs and $t$ functions $h_i$.
\end{definition}
\begin{theorem}[\cite{erdos1975sparse,Valiant,C08,V09}]\label{thm:valiant}
Let $f\colon\zo^{n}\to\zo^{n}$. For every $c, \eps>0$ there exists $\delta>0$ such that
\begin{enumerate}
\item If $f$ can be computed by a circuit of size $cn$ and depth $c\log{n}$, then $f$ has an $(\delta n/\log\log{n},n^\eps)$-solution in the common bits model.
\item If $f$ can be computed by a circuit of size $cn$ and depth $c\log{n}$, then $f$ has an $(\eps n, 2^{\log{n}^{1-\delta}})$-solution in the common bits model.
\item If $f$ can be computed by a series-parallel circuit of size $cn$ (and unbounded depth), then $f$ has an $(\eps n, \delta)$-solution in the common bits model.
\end{enumerate}
\end{theorem}

Now we show that a weak OWF in the random oracle model with preprocessing (for certain settings of parameters) implies a super-linear circuit lower bound. This proof employs the approach used in~\cite{dgw19,viola2018lower,cgk18,ramamoorthy2019equivalence}.
For ease of exposition, in the next theorem we assume that the preprocessing is measured in bits (i.e., the word size $w$ is a constant number of bits). For this reason, a trivial inverter for a function $f^R:\zo^{n'}\to\zo^{n'}$ requires space $n'2^{n'}$. This assumption is not crucial, and the result easily generalizes to any $w$, in which case the amount of preprocessing is decreased by a factor of $w$.
\begin{theorem}
\label{thm:circuits}
Let $f^R:\zo^{n'}\to\zo^{n'}$ be a $(S, T,\eps)$-hard OWF in the random oracle model with preprocessing,
			for a random oracle  $R:\zo^{n}\to\zo^{n}$, where $n'=O(n)$. We construct a function $G\in\mathbf{P}$ such that:
\begin{enumerate}
\item If $S\geq \omega\left(\frac{n'2^{n'}}{\log{n'}}\right), T\geq2^{\delta n}$ and $\eps=1$ for a constant $\delta>0$, then $G$ cannot be computed by a circuit of linear size and logarithmic depth.
\item If $S\geq \delta n' 2^{n'}, T\geq2^{n^{1-o(1)}}$ and $\eps=1$ for a constant $\delta>0$, then $G$ cannot be computed by a circuit of linear size and logarithmic depth.
\item If $S\geq \delta n' 2^{n'}, T\geq\omega(1)$ and $\eps=1$ for a constant $\delta>0$, then $G$ cannot be computed by a series-parallel circuit of linear size.
\end{enumerate}
\end{theorem}
\begin{proof}
Let $\tilde{N}=n2^n$, and let
$g:\zo^{\tilde{N}}\times\zo^{n'}\to\zo^{n'}$ be defined as
\[
g(d,q)=\min\{y\in \zo^{n'} ~:~ f^d(y)=q\} \ .
\]
Let $\ell\coloneqq \frac{n'2^{n'}}{\tilde{N}}$, and let us define $\ell$ data structure problems $g_i\colon\zo^{\tilde{N}}\times [\tilde{N}/n']\to\zo^{n'}$ for $i\in[\ell]$ as follows:
\[
g_i(d,q) = g(d,q+(i-1)\cdot\tilde{N}/n') \ ,
\]
where we identify a binary string from $\zo^{n'}$ with an integer from $[2^{n'}]$. Finally, we define $G\colon\zo^{\tilde{N}+\log{\ell}}\to\zo^{\tilde{N}}$ as
\[
G(d, i) = g_i(d,1) \| \ldots \| g_i(d,\tilde{N}/n')
\ .
\]
We claim that $G$ cannot be computed by a circuit of linear size and logarithmic depth (a series-parallel circuit of linear size, respectively). The proofs of the three statements of this theorem follow the same pattern, so we only present the proof of the first one.

Assume, for contradiction, that there is a circuit of size $O(\tilde{N})$ and depth $O(\log{\tilde{N}})$ that computes $G$. By Theorem~\ref{thm:valiant}, $G$ has an $(s, t)$-solution in the common bits model, where $s=O(\tilde{N}/\log\log{\tilde{N}})=O(2^n n/\log{n})$ and $t=\tilde{N}^{\delta/2}<2^{\delta n}$. Since each output of $g$ is a part of the output of $G(\cdot,i)$ for one of the $\ell$ values of $i$, we have that $g$ has an $(s\cdot\ell, t)$-solution in the common bits model. In particular, $g$ can be computed with preprocessing $s\cdot\ell=O(n'2^{n'}/\log{n})$ and $t=2^{\delta n}$ queries to the input. This, in turn, implies a $(n'2^{n'}/\log{n}, 2^{\delta n})$-inverter for $f^R$.

Finally, we observe that the function $G$ can be computed by $\tilde{N}/n'$ evaluations of $g$, and $g$ is trivially computable in time $2^{n'}\cdot\poly(n')$. Therefore, $G\in {\bf DTIME}[\tilde{N}\cdot 2^{n'}]={\bf DTIME}[2^{O(n)}]={\bf DTIME}[\tilde{N}^{O(1)}]={\bf P}$.
\end{proof}
\begin{remark}
We remark that $\eps=1$ is the strongest form of the theorem, \emph{i.e.}, the premise of the theorem only requires a function $f^R$ which cannot be inverted on \emph{all} inputs. Also, it suffices to have $f^R$ which cannot be inverted by \emph{non-adaptive} algorithms, \emph{i.e.}, algorithms where $\cA_2$ is non-adaptive (see Definition~\ref{def:our-model}). 
\end{remark}

\section*{Acknowledgments} Many thanks to Erik Demaine for sending us
a manuscript of his one-query lower bound with Salil Vadhan~\cite{DV01}. We
also thank Henry Corrigan-Gibbs and Dima Kogan for useful discussions.

The work of AG is supported by a Rabin Postdoctoral Fellowship. The work of TH
is supported in part by the  National  Science Foundation  under grants CAREER
IIS-1149662, CNS-1237235 and CCF-1763299, by the Office of Naval Research under
grants YIP N00014-14-1-0485 and N00014-17-1-2131, and by a Google Research
Award. 
The work of SP is supported by the MIT Media Lab's Digital Currency Initiative,
and its funders; and an earlier stage of SP's work was funded 
by the following grants: NSF MACS (CNS-1413920), 
DARPA IBM (W911NF-15-C-0236), Simons Investigator award agreement dated June 5th, 2012,
and the Center for Science of Information (CSoI), an NSF Science and Technology Center, 
under grant agreement CCF-0939370.
The work of VV is supported in part by NSF Grants CNS-1350619, CNS-1718161
and CNS-1414119, an MIT-IBM grant, a Microsoft Faculty Fellowship and
a DARPA Young Faculty Award.

\bibliographystyle{alpha}
\bibliography{bib}

\end{document}